\documentclass[aps,prx,reprint,superscriptaddress,reprint,amsmath,amssymb]{revtex4-2}

\usepackage{graphicx}%
\usepackage{dcolumn}%
\usepackage{bm}%
\usepackage{latexsym}
\usepackage{array}
\usepackage{setspace}

\usepackage{multirow}

\usepackage{xcolor}
\usepackage[normalem]{ulem}

\usepackage[justification=RaggedRight]{caption}

\def\sint{\ifmmode{- \!\!\!\!\!\! \int}
    \else{\hbox{$- \!\!\!\! \int \ $}}\fi}

\newcommand{\tr}{\mbox{tr}}

\usepackage{amsthm}
\newtheorem{theorem}{Theorem}
\newtheorem{lemma}{Lemma}
\newtheorem{corollary}[theorem]{Corollary}

\usepackage{apptools}
\AtAppendix{\counterwithin{lemma}{section}}
\AtAppendix{\counterwithin{theorem}{section}}

\usepackage{algpseudocode}

\usepackage[justification=centerlast]{caption}
\usepackage[labelformat=simple]{subcaption}

\begin{document}

\title{Deterministic 
	entanglement distribution on %
	series-parallel quantum networks%
}%

\author{Xiangyi~Meng}%
\affiliation{Network Science Institute and Department of Physics, Northeastern University, Boston, Massachusetts 02115, USA}%
\author{Yulong~Cui}%
\affiliation{State Key Laboratory of Industrial Control Technology, Zhejiang University, Hangzhou, Zhejiang 310027, China}%
\author{Jianxi~Gao}%
\affiliation{Department of Computer Science, Rensselaer Polytechnic Institute, Troy, New York 12180, USA}
\author{Shlomo~Havlin}
\affiliation{Department of Physics, Bar-Ilan University, 52900 Ramat-Gan, Israel}%
\affiliation{Department of Physics, Boston University, Boston, Massachusetts 02215, USA}%
\author{Andrei~E.~Ruckenstein}
\email{andreir@bu.edu}
\affiliation{Department of Physics, Boston University, Boston, Massachusetts 02215, USA}%

\date{\today}

\begin{abstract}
	The performance of distributing entanglement between two distant nodes in a large-scale quantum network (QN) of partially entangled bipartite pure states is generally benchmarked against the classical entanglement percolation (CEP) scheme. Improvements beyond CEP were only achieved by nonscalable strategies for restricted QN topologies. 
	This paper explores and amplifies a new and more effective mapping of a QN, referred to as concurrence percolation theory (ConPT), that suggests using {deterministic} rather than probabilistic protocols for scalably improving on CEP across arbitrary QN topologies.
	More precisely, we implement ConPT via a {deterministic entanglement transmission} (DET) scheme that is fully analogous to resistor network analysis, with the corresponding series and parallel rules represented by deterministic entanglement swapping and concentration protocols, respectively. 
{The main contribution of this paper is to establish a powerful mathematical framework, which is applicable to arbitrary $d$-dimensional information carriers (qudits), that provides different natural optimality metrics in terms of generalized $k$-concurrences (a family of fundamental entanglement measures) for different QN topology. In particular, we conclude that the introduced DET scheme (a) is optimal over the well-known nested repeater protocol  for distilling entanglement from partially entangled qubits and (b) leads to higher success probabilities of obtaining a maximally entangled state than using CEP. The implementation of the DET scheme is experimentally feasible as tested on IBM's quantum computation platform.}
\end{abstract}

\maketitle

\section{Introduction}

Most quantum communication technologies find their roots in entanglement, a primordial quantum resource~\cite{q-resour_cg19}.
Yet, the nonlocality of entanglement facilitating quantum communication becomes unrealistic at large space scales where faithful quantum information is submerged in thermal incoherent noise~\cite{q-dissipative-syst}, making it very inefficient to directly establish point-to-point long-distance entanglement.
To overcome this limit of classical locality and %
{indirectly} distribute entanglement among arbitrary parties requires utilizing some intermediate parties as ``relays'' and applying to them a variety of fundamental communication protocols, including entanglement swapping~\cite{entangle-swap_zzhe93,QEP-series-rule_bvk99} and entanglement concentration~\cite{entangle-conc_bbps96,entangle-conc_mk01},
as well as composite protocols that are usually recursive and more complex, such as the nested repeater protocol~\cite{q-repeater_bdcz98,q-repeater_ssrg11,q-repeater-graph-state_bbe17}, etc.

Since all parties are physically separated from one another so that locality holds, only local operations and classical communication (LOCC) can be applied to the information
carriers (e.g.,~qubits) of different parties. Such consideration allows us to treat each party as a \emph{node} that consists of a collection of information carriers and treat each partially entangled bipartite (multipartite) state formed by internode information carriers as a 
\emph{link} (hyperlink), essentially giving rise to a network representation of the structure of locality of entanglement resources, widely known as the \emph{quantum network} (QN)~\cite{QEP_acl07}. 
A number of QN-based protocols have since been introduced, e.g.,~$q$-swapping~\cite{QEP-q-swap_cc09} and path routing~\cite{q-netw-route_p19}, combining communication protocols with network science~\cite{q-netw-summ_bfd19}. In practice, basic communication protocols for small-scale QNs have been experimentally demonstrated using diamond nitrogen-vacancy centers~\cite{q-netw_krhbknbtmh17,q-netw_hkmsvtmh18,q-netw_phbbhsvtmdwh21} and ion traps~\cite{entangle-swap_rmkvschb08,q-netw-ion-trap_ahprshde09}. %
It is believed that practical QNs are easier to {scale} than universal quantum computing platforms, since the nonlocality of information carriers only needs to be maintained within each node and not across the global QN platform. %
The future of a worldwide large-scale QN, namely, a ``quantum Internet''~\cite{q-internet_k08}, is hence foreseeable. This, on the other hand, urgently demands a more efficient design of QN-based communication protocols that should be \emph{scalable} with network size, be \emph{adaptable} to different network topologies, and align with well-developed methods in complex network analysis~\cite{stat-of-netw_ab02,struct-dyn-netw}.

The particular task of entanglement distribution~\cite{entangle-distrib_czkm97} between \emph{two} arbitrary nodes, which we coin as {entanglement transmission} in the QN context, is of special interest. %
The discovery of an exact mapping between entanglement transmission and classical bond percolation theory~\cite{percolation-theor_e80} gives rise to a straightforward entanglement transmission scheme, called classical entanglement percolation (CEP)~\cite{QEP_acl07}, that greatly reduces the task of distributing a singlet (i.e.,%
~a pair of maximally entangled qubits) between two nodes to a pure percolation problem. The CEP scheme is naturally scalable and adaptable for arbitrary QN. However, immediately after it was introduced, {it was realized that the CEP does not give the {optimal} success probability of obtaining such a singlet~\cite{QEP_acl07}.} The problem of calculating the optimal probability was later proved to be extremely complicated even for a one-dimensional (1D) chain~\cite{QEP-detail_pcalw08,QEP-1d_s16}. %
Even just improving this success probability for large-scale QN can only be systematically done for special network topology with a few other limitations~\cite{QEP-detail_pcalw08}, making it difficult to find a better entanglement transmission scheme than CEP that is scalable and adaptable in the same satisfactory way. %

This difficulty eventually leads to the questioning of the exclusivity of the classical-percolation-theory-based mapping itself: It was only recently realized that another mapping from entanglement transmission to a new statistical theory called concurrence percolation theory (ConPT)~\cite{conpt_mgh21,conpt_mmhksg22} exists, bearing its name due to its analogy to percolation theory yet rooted not in probability but in concurrence, a fundamental measure of bipartite entanglement~\cite{concurrence_hw97,k-concurrence_g05}. The ConPT predicts a lower entanglement transmission threshold than the classical percolation threshold, exhibiting a ``quantum advantage'' that is purely structural, independent of nontopological details. To achieve the ConPT threshold, it has been suggested that \emph{deterministic} communication protocols that only produce outcomes with certainty (up to unitary equivalence) should be used, so that the protocols can be applied recursively without becoming involved with probabilistic distributions~\cite{conpt_mgh21}.

\begin{figure*}[t!]
	\centering
	\includegraphics[width=5.514in]{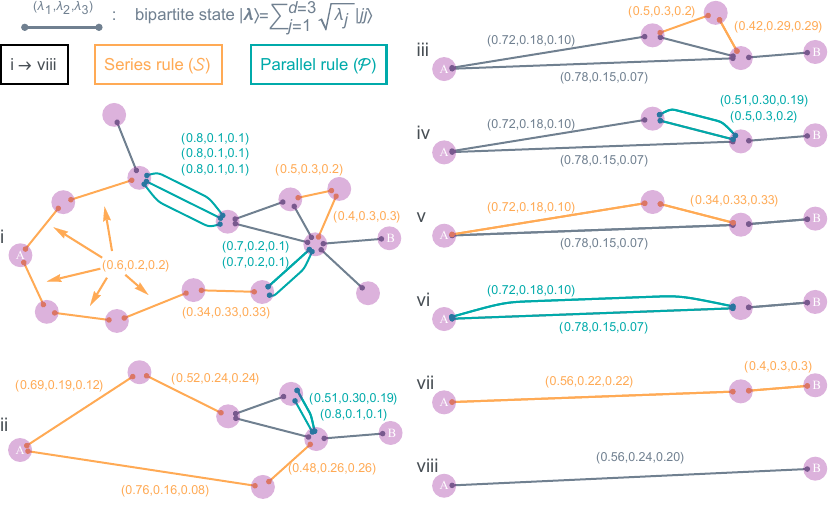}
	\caption{\label{fig_det-example}
		Demonstration of the deterministic entanglement transmission (DET) scheme for arbitrary series-parallel quantum network (QN) with $d$-dimensional information carriers [e.g.,~qutrits ($d=3$)]. Starting from an initial QN (i), a final pure state (viii) can be deterministically produced between Alice (A) and Bob (B) by {canonically applying the series and parallel rules (Table~\ref{table_qdet-rules}).}
		\hfill\hfill}
\end{figure*}

In light of this, here we formalize a specific {deterministic entanglement transmission} (DET) scheme  {(Sec.~\ref{sec_qdet})} that works for two arbitrary nodes in a QN, provided that the network topology between the two nodes is series-parallel, 
i.e.,~can be decomposed into only two connectivity rules---\emph{series and parallel rules}---in accordance with resistor network analysis. This DET scheme is built by explicitly expressing the series and parallel rules as two functions---a swapping function $\mathcal{S}$ and a concentration function $\mathcal{P}$---which can be implemented by two fundamental communication protocols: \emph{entanglement swapping and concentration}, respectively.
{This is possible since, while both protocols are generically probabilistic, they can also be implemented deterministically.} Our scheme has two features:
\begin{enumerate}
	\item The DET scheme is fully scalable and adaptable for arbitrary series-parallel QN, as it {reduces} the entanglement transmission task to a problem of calculating path connectivity between two arbitrary nodes, fully analogous to calculating the total resistance in a series-parallel resistor network (Fig.~\ref{fig_det-example}).
	\item  The DET scheme is defined for $d$-dimensional {qudits (qutrits~\cite{qutrit-scramble_brsobkdmyys21}, ququarts~\cite{ququart-superdense_hglhlg18}, etc.)} in general, making the DET an all-purpose scheme for more exotic design of quantum information devices.
\end{enumerate}
{Unlike the CEP scheme, the DET exhibits different levels of \emph{optimality} {(Sec.~\ref{sec_optim})}, %
given in terms of a family of $d$ concurrence monotones~\cite{k-concurrence_g05}. Using majorization theory, we calculate inequalities for the $\mathcal{S}$ and $\mathcal{P}$ functions and their recursive combinations, finding that the outcome of DET is optimal (Table~\ref{table_qdet-optimal}) 
over that of probabilistic schemes.
The powerful mathematical treatment also allows us to show that, if the deterministic outcome of DET is further converted to a maximally entangled qudit state, then the success probability of doing so
is always higher than the CEP result.
}

\begin{table}[h]
	\centering 
	\caption{\label{table_qdet-optimal}Optimality of DET in terms of $k$-concurrences $C_k$ (Appendix~\ref{sec_concurrence}) with different QN topologies (Fig.~\ref{fig_seri-para-qn}).~\hfill\hfill}
	\begin{tabular}{r|l}
		\hline\hline
		Network topology & DET optimizes...\\
		\hline
		Simple series & Average $C_d$\\
		Simple parallel & Average $C_k$ ($1<k\le d$)\\
		Parallel-then-series & Average $C_d$\\
		Series-then-parallel & Worst-case $C_d$ ($d=2$ only)\\
		Series-parallel & Worst-case $C_d$ ($d=2$ only)\\
		\hline\hline
	\end{tabular}
\end{table}

{ %
Our results also suggest that}
the well-known {nested quantum repeater protocol}~\cite{q-repeater_bdcz98} introduced for a 1D chain of parties with multiple bipartite partially entangled qubits shared in between---which is essentially a parallel-then-series QN---is \emph{not} a good strategy for optimizing concurrence. %
Instead of applying entanglement swapping and concentration in a nested way, we argue that the best strategy is to apply entanglement concentration {once and for all} between every two adjacent nodes and then perform entanglement swapping along the concentrated links.
{The proof of this result} relies on a special reverse arithmetic-mean--geometric-mean (AM--GM) inequality (which we will prove in Appendix~\ref{sec_reverse}). The lack of a necessity for nesting could greatly simplify communication protocol design and make DET preferable practically.

{
For demonstration, we show that both the series and parallel rules are experimentally feasible {(Sec.~\ref{sec_implement})}, typically performing with fidelities of $92.4\%$ and $78.2\%$, respectively, tested on IBM's quantum computation platform \emph{Qiskit}~\cite{qiskit_a21}. The performance could be further improved considering recent breakthroughs in raising the fidelity of two-qubit gates above $99\%$~\cite{q-cnot-fidel_cbdsbe19,q-cnot-fidel_kwsmckkdm21}.
}

\section{Quantum Network}

In this paper, we focus on a pure-state version of a QN where each link is a $d^2$-dimensional bipartite pure state shared by the two connected nodes~\cite{QEP_acl07}. Such a bipartite pure state can always be written as $\sum_{\mu,\nu}\Psi_{\mu\nu}\left|\mu\nu\right\rangle\in\mathbb{C}^{d\times d}$, which allows {matricization}, i.e.,~allows it to be mapped to a $d\times d$ matrix $\boldsymbol{\Psi}$ with elements $\Psi_{\mu\nu}$, $\mu,\nu=1,2,\cdots,d$.
Left or right multiplication of $\boldsymbol{\Psi}$ by a unitary matrix corresponds to a local unitary transformation performed by either of the two parties, respectively. Thus, up to unitary equivalence, the state can always be locally transformed into a diagonal form $\left|\boldsymbol{\lambda}\right\rangle=\sum_{j=1}^{d}\sqrt{\lambda_{j}}\left|jj\right\rangle$ by a singular value decomposition, so that each link is exclusively represented by $d$ positive Schmidt numbers, $\boldsymbol{\lambda}\equiv\left(\lambda_1,\lambda_2,\cdots,\lambda_d\right)$. Additionally, note that $\boldsymbol{\lambda}$ is always subject to the normalization constraint %
$\sum_{j=1}^{d}\lambda_{j}= 1$. %
Thus, when $d=2$ (qubits), each link can be represented by only one parameter~\cite{conpt_mgh21}.

This representation by a sequence of Schmidt numbers is effective, since singular values are equipped with a built-in majorization {preorder} (Appendix~\ref{sec_majorization}). This obviously only works for {bipartite} states. Tripartite states, for example, although they can be completely classified by five parameters including one phase and four moduli~\cite{tripartit-entangle_aacjlt00}, are not compatible with such a preorder, thus manifesting as a fundamentally difficult \emph{``(quantum) three-body problem.''}
Note that traditionally, each link in a network is also a ``bipartite'' pairwise notation, always connecting two nodes. Hence it seems a perfect match to study the statistics of a large number of bipartite states using a ``bipartite'' network theory, namely, a\emph{``statistical theory of (quantum) two-body problems.''} %
In contrast, a ``complex'' QN consisting of multipartite entangled states as ``hyperlinks'' would be beyond the scope of this paper, since the difficulties coming from both quantum many-body theory and hypergraph theory are actually twofold.

\subsection{Quantum Communication Protocols}
\label{sec_protocol}

{We briefly review two LOCC protocols for bipartite states that are often used in quantum communication:}

\subsubsection{Entanglement swapping}

{The motivation for the entanglement swapping protocol is as follows: Suppose there are three parties, Alice--Relay--Bob (A--R--B), who share two bipartite states, $\left|\boldsymbol{\lambda}_a\right\rangle=\sum_{j=1}^{d}\sqrt{\lambda_{a,j}}\left|jj\right\rangle_\text{AR}$ (shared between A--R) and $\left|\boldsymbol{\lambda}_b\right\rangle=\sum_{k=1}^{d}\sqrt{\lambda_{b,k}}\left|kk\right\rangle_\text{RB}$ (shared between R--B), respectively. %
Since A and B do not directly share any entanglement, a swapping protocol~\cite{QEP-series-rule_bvk99} must be performed, which requires a set of quantum measurements to be conducted on R. The result, which is probabilistic in general, is an ensemble of bipartite states directly shared between A and B, entangled.}

To be specific, if we matricize $\left|\boldsymbol{\lambda}_a\right\rangle$ and $\left|\boldsymbol{\lambda}_b\right\rangle$ and denote them by two diagonal matrices $\text{diag}(\boldsymbol{\lambda}_a)^{1/2}$ and $\text{diag}(\boldsymbol{\lambda}_b)^{1/2}$, then after R performs the measurements, the resultant (unnormalized) state between A and B can be written in the matrix form,
\begin{equation}
\label{eq_swap}
\boldsymbol{\Psi}^{\alpha}=\text{diag}(\boldsymbol{\lambda}_a)^{1/2} \mathbf{X}^{\alpha} \text{diag}(\boldsymbol{\lambda}_b)^{1/2}
\end{equation} 
for outcome $\alpha$, subject to probability $p_\alpha=\left\|\boldsymbol{\Psi}^{\alpha}\right\|_2^2$ of the outcome~\cite{q-netw-summ_pjcla13}. The quantum measurements are encoded by a set of arbitrary matrices $\{\mathbf{X}^{\alpha}\}$ constrained by the completeness relation, 
\begin{equation}
\label{eq_povm}
\sum_{\alpha}\left(X^{\alpha}_{\mu\nu}\right)^* X^{\alpha}_{\mu'\nu'}=\delta_{\mu\mu'}\delta_{\nu\nu'},%
\end{equation} 
which is sufficient for $\{\mathbf{X}^{\alpha}\}$ to denote a positive operator-valued measure (POVM).

\subsubsection{Entanglement concentration}
Given $\left|\boldsymbol{\lambda}\right\rangle=\sum_{j=1}^{m}\sqrt{\lambda_{j}}\left|jj\right\rangle$ for $m\in \mathbb{Z}^+$, what is the maximum probability of converting it to a new state $\left|\mathbf{x}\right\rangle=\sum_{j=1}^{m}\sqrt{x_{j}}\left|jj\right\rangle$ by LOCC? Such a probability %
is well known~\cite{QEP-parallel-rule_v99}:
\begin{equation}
\label{eq_concentrate}
p_{\boldsymbol{\lambda}\to\mathbf{x}}=\min_{0\le k< m}\frac{1-\sum_{j=1}^{k}\lambda^{\downarrow}_j}{1-\sum_{j=1}^{k}x^{\downarrow}_j}.
\end{equation}
In particular, note that it is possible to obtain a new state $\left|\mathbf{x}\right\rangle$ of a different dimension $d$ if $d \le m$, since one can pad $\mathbf{x}$ by zeros and build a new sequence of $m$ Schmidt numbers, $\mathbf{x}'=\mathbf{x}\oplus(\overbrace{0,0,\cdots,0}^{m-d})$, which when inserted into Eq.~\eqref{eq_concentrate} yields the maximum probability $p_{\boldsymbol{\lambda}\to\mathbf{x}}$~\cite{monotone_v00}.
{The motivation of the entanglement concentration protocol is that, by treating multiple partially entangled bipartite states together as one tensor-product state of a much larger dimension $m$, one can convert such a tensor-product state to a bipartite state of fewer dimensions $d$ but with potentially higher entanglement than each previously individual bipartite state.} 
{A concentration protocol can be either probabilistic or deterministic, for $p_{\boldsymbol{\lambda}\to\mathbf{x}}=1$ is clearly possible by Eq.~\eqref{eq_concentrate}. %
}

\subsection{Entanglement Transmission Schemes}
\label{sec_scheme}
A \emph{scheme} is a (possibly infinite) number of protocols that can be collectively applied to a QN. One of the well known entanglement transmission schemes is the classical entanglement percolation (CEP) scheme for $d=2$. The CEP scheme consists of two steps~\cite{QEP_acl07}: (i) Use the concentration protocol to convert each link to a singlet, the probability of which is given by Eq.~\eqref{eq_concentrate}. (ii) Find a path of links connecting two arbitrary nodes A and B that have all been successfully converted to singlets, and apply the swapping protocol to them one by one. The final state between A and B is guaranteed to be a singlet by Eq.~\eqref{eq_swap}. The probability of finding such a path can be considered as a measure of the entanglement transmission ability between A and B, determined by the network topology. This probability, however, is known to be nonoptimal, since for some special topology, e.g.,~a honeycomb with double links between every two connected nodes, adding a ``preprocessing'' step of entanglement swapping to the CEP scheme may change the network topology to a triangular type and potentially increase the probability of establishing a singlet between A and B~\cite{QEP_acl07}.

A deterministic entanglement transmission (DET) scheme, on the other hand, does not use generic probabilistic protocols to increase the singlet conversion probability, but uses {deterministic} protocols to produce a single highly entangled state between A and B (demonstrated in Fig.~\ref{fig_det-example}). Note that after performing each deterministic protocol, the new system is still an intact QN, i.e.,~a pure state as a whole, not a probabilistic ensemble. This makes DET {purely ``nonclassical'' (without mixing with probability measures) and} easily scalable with network size.

The specific DET scheme we introduce, which utilizes only deterministic swapping and concentration protocols, can be applied to arbitrary series-parallel QN.

\section{A Deterministic Entanglement Transmission Scheme}
\label{sec_qdet}

\subsection{Series-Parallel Network Topology}

\begin{figure}[t]
	\centering
	\begin{minipage}[b]{1.65in}	
		\centering
		{\includegraphics[width=1.65in]{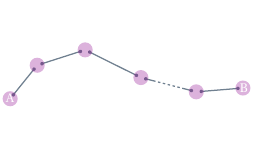}
			\vspace{-6mm}\subcaption{Simple series QN.\label{fig_seri}}}
	\end{minipage}
	\begin{minipage}[b]{1.65in}
		{\includegraphics[width=1.65in]{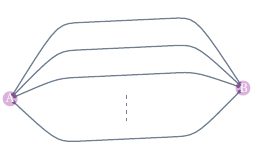}
			\vspace{-6mm}\subcaption{Simple parallel QN.\label{fig_para}}}
	\end{minipage}

	\begin{minipage}[b]{1.65in}	
		\centering
		{\includegraphics[width=1.65in]{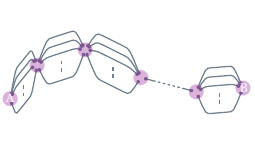}
			\vspace{-6mm}\subcaption{Parallel-then-series QN.\label{fig_para-then-seri}}}
	\end{minipage}
	\begin{minipage}[b]{1.65in}
		{\includegraphics[width=1.65in]{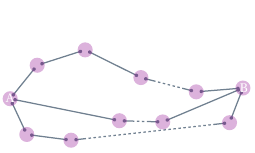}
			\vspace{-6mm}\subcaption{Series-then-parallel QN.\label{fig_seri-then-para}}}
	\end{minipage}

	\begin{minipage}[b]{1.65in}	
		\centering
		{\includegraphics[width=1.65in]{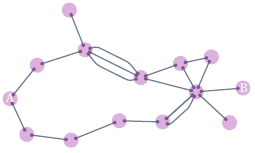}
			\vspace{-6mm}\subcaption{Series-parallel QN.\label{fig_seri-para}}}
	\end{minipage}
	\begin{minipage}[b]{1.65in}
		{\includegraphics[width=1.65in]{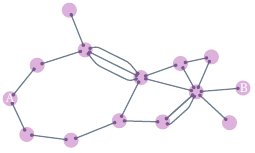}
			\vspace{-6mm}\subcaption{General QN.\label{fig_general}}}
	\end{minipage}	
	\caption{\label{fig_seri-para-qn}Different QN topologies between A and B. 
		\hfill\hfill}
\end{figure}

Whether a network is series-parallel or not depends on which two nodes of interest are chosen~\cite{series-parallel-netw_d65}. Given two nodes A and B, %
the network topology can be characterized into different categories (Fig.~\ref{fig_seri-para-qn}). All topologies between A and B given in Figs.~\ref{fig_seri}--\ref{fig_seri-para} are considered series-parallel, but the topology in Fig.~\ref{fig_general} is not, due to an existing ``bridge'' compared with Fig.~\ref{fig_seri-para}. %
It is worth noting that most known realistic complex networks can be approximated as series-parallel, since cycles can usually be ignored in infinite-dimensional systems by the Bethe approximation~\cite{netw-percolation_ceah00}. %

\subsection{Series and Parallel Rules}

\begin{table}[h]
	\centering 
	\caption{\label{table_qdet-rules}Series and parallel rules for DET.\hfill\hfill}
	\begin{tabular}{p{2.5cm}| p{5.7cm}}
		\hline\hline
		\textbf{S}eries rule\newline (Swapping) & %
		$\boldsymbol{\lambda}=\boldsymbol{\lambda}_1\texttildelow \mathcal{S} \texttildelow \boldsymbol{\lambda}_2 \texttildelow \mathcal{S} \texttildelow \cdots \texttildelow \mathcal{S} \texttildelow \boldsymbol{\lambda}_n$,\newline 
		where $\mathcal{S}:\left(\mathbb{R}_{+}^{d},\mathbb{R}_{+}^{d}\right)\to\mathbb{R}_{+}^d$%
		~[Eq.~\eqref{eq_s}]
		\\
		\hline
		\textbf{P}arallel rule\newline (Concentration) &  $\boldsymbol{\lambda}=\mathcal{P}(\boldsymbol{\lambda}_1\otimes \boldsymbol{\lambda}_2\otimes\cdots\otimes \boldsymbol{\lambda}_n)$,\newline 
		where {{$\mathcal{P}:\mathbb{R}_{+}^{d^n}\to\mathbb{R}_{+}^d$}}%
		~[Eq.~\eqref{eq_p}]\\
		\hline\hline
	\end{tabular}
\end{table}

The series and parallel rules of our DET scheme are given in Table~\ref{table_qdet-rules}, where $\otimes$ stands for the Kronecker product, and, for simplicity, the infix notation $\texttildelow$ is introduced for representing binary operations, i.e.,~$x\texttildelow f \texttildelow y \equiv f(x,y)$, which is left grouping, i.e.,~$x\texttildelow f \texttildelow y \texttildelow f \texttildelow z \equiv (x\texttildelow f \texttildelow y) \texttildelow f \texttildelow z$.
{The two functions, $\mathcal{S}(\mathbf{x},\mathbf{y})$ and $\mathcal{P}(\mathbf{x})$, as well as the motivations for introducing them  are given as follows.}

\subsubsection{Swapping function $\mathcal{S}(\mathbf{x},\mathbf{y})$}
\emph{Definition.---}We define %
$\mathcal{S}:\left(\mathbb{R}_{+}^{d},\mathbb{R}_{+}^{d}\right)\to\mathbb{R}_{+}^d$ by
\begin{eqnarray}
\label{eq_s}
\mathcal{S}(\mathbf{x},\mathbf{y})=d\times \sigma^2(\text{diag}(\mathbf{x}^{\downarrow})^{1/2}\mathbf{V}\text{diag}(\mathbf{y}^{\downarrow})^{1/2})
\end{eqnarray}
where %
$\sigma(\boldsymbol{\Psi})$ denotes the singular values 
of matrix $\boldsymbol{\Psi}$ and $\sigma^2(\boldsymbol{\Psi})$ denotes the entry-wise square of $\sigma(\boldsymbol{\Psi})$, both arranged in descending order so that $\mathcal{S}(\mathbf{x},\mathbf{y})\equiv\left[\mathcal{S}(\mathbf{x},\mathbf{y})\right]^{\downarrow}$.
The matrix $\mathbf{V}$ is constant and unitary, with elements $V_{\mu\nu}=d^{-1/2}\exp\left(-2\pi i\mu\nu/d\right)$, $\mu,\nu=1,2,\cdots,d$.

\emph{Properties.---}$\mathcal{S}(\mathbf{x},\mathbf{y})$ has the following properties.
\begin{itemize}
	\item {Permutation invariance [$\mathbf{x}\to\text{perm}(\mathbf{x})$]:}
	\begin{equation}
	\label{eq_s_perm}
	\mathcal{S}(\mathbf{x},\mathbf{y})=\mathcal{S}(\text{perm}(\mathbf{x}),\mathbf{y}).
	\end{equation}
	\item {Trace preserving:}
	\begin{equation}
	\label{eq_s_trace}
	\tr(\mathcal{S}(\mathbf{x},\mathbf{y}))=\tr(\mathbf{x}) \tr(\mathbf{y}).
	\end{equation}
	\item {Isotone:}
	\begin{equation}
	\label{eq_s_isotone}
	\mathcal{S}(\mathbf{x},\mathbf{z})\succ \mathcal{S}(\mathbf{y},\mathbf{z}),\forall \mathbf{z} \text{ if } \mathbf{x}\succ \mathbf{y}.
	\end{equation}
	\item {Commutativity:}
	\begin{equation}
	\label{eq_s_commut}
	\mathcal{S}(\mathbf{x},\mathbf{y})=\mathcal{S}(\mathbf{y},\mathbf{x}).
	\end{equation}	
	\item  {Associativity:}
	\begin{equation}
	\label{eq_s_assoc}
	\mathcal{S}(\mathcal{S}(\mathbf{x},\mathbf{y}),\mathbf{z})= \mathcal{S}(\mathbf{x},\mathcal{S}(\mathbf{y},\mathbf{z})) \text{ for } d \le 3.
	\end{equation}
\end{itemize}

Equations~\eqref{eq_s_perm}~and~\eqref{eq_s_commut} hold because $\sigma(\boldsymbol{\Psi})$ is invariant under Hermitian conjugate and unitary transformation of $\boldsymbol{\Psi}$.
Equation~\eqref{eq_s_isotone} holds as a result of Theorem~\ref{theorem_convex-perm-isotone}, given the facts that $\mathcal{S}$ is permutation invariant [Eq.~\eqref{eq_s_perm}] and convex (Lemma~\ref{theorem_s_convex}). 
Equation~\eqref{eq_s_trace} can be inferred from its validity for the special case $\tr(\mathbf{x})= \tr(\mathbf{y})=1$. %

In particular, Eq.~\eqref{eq_s_assoc} is valid for $d\le 3$ because both its left-hand side (LHS) and right-hand side (RHS) have two constraints: trace preserving [Eq.~\eqref{eq_s_trace}] and determinant preserving (Lemma~\ref{theorem_s_det}). Additionally, employing the duality (Lemma~\ref{theorem_s_adj}) produces two more constraints in terms of ``dual'' trace preserving and ``dual'' determinant preserving. The latter, however, can be shown to be equivalent to the original determinant preserving constraint. Therefore, there are three independent constraints in total, and thus the LHS and RHS must be equal when $d\le 3$. When $d> 3$, counterexamples can be easily found.

\emph{Implementability by LOCC.---}It remains to show that $\mathcal{S}(\boldsymbol{\lambda}_a,\boldsymbol{\lambda}_b)$ can be implemented by the swapping protocol [Eq.~\eqref{eq_swap}] in a deterministic manner. 
In short,  given a special set of matrices $\{\mathbf{X}^{\alpha}\}$ each with elements $X^{\alpha}_{\mu\nu}=d^{-1}e^{-\alpha \left(d \mu +\nu\right)2\pi i/d^2-2\pi i \mu \nu/d}$ for $\alpha=1,2,\cdots,d^2$, one can verify that $\{\mathbf{X}^{\alpha}\}$ satisfies Eq.~\eqref{eq_povm}, thus denoting a POVM. Hence, performing the swapping protocol and following Eq.~\eqref{eq_swap}, the resultant (unnormalized) state $\boldsymbol{\Psi}^{\alpha}$ is obtained, for each outcome $\alpha$, with elements $\Psi^{\alpha}_{jk}=d^{-1}\sqrt{\lambda^{\downarrow}_{a,j}\lambda^{\downarrow}_{b,k}}e^{-\left(\alpha d j +\alpha k + d j k\right)2\pi i/d^2}$.
Now, as long as Alice and Bob are shared by the Relay the classical information of which $\alpha$ is obtained, they can always rotate their shared state $\boldsymbol{\Psi}^{\alpha}$ by some phase accordingly and transform it locally into a new $\alpha$-independent state $\boldsymbol{\Psi}'$ with elements $\Psi'_{jk}=d^{-1}\sqrt{\lambda^{\downarrow}_{a,j}\lambda^{\downarrow}_{b,k}}e^{-\left(2\pi i\right)jk/d} $, the Schmidt numbers of which are then given by $\mathcal{S}(\boldsymbol{\lambda}_a,\boldsymbol{\lambda}_b)$~[Eq.~\eqref{eq_s}] after normalization of $\boldsymbol{\Psi}'$. The independence of $\boldsymbol{\Psi}'$ on $\alpha$ confirms that the protocol is deterministic.
{An example of the full construction of $\mathcal{S}(\boldsymbol{\lambda}_a,\boldsymbol{\lambda}_b)$ is given in Sec.~\ref{sec_implement}.}

{The construction above was first given in Ref.~\cite{k-concurrence_g05}, behind which the motivation is that under such a special choice of $\{\mathbf{X}^{\alpha}\}$, the swapping protocol actually optimizes a determinant-based entanglement measure ($G$-concurrence) of the final state. This is made manifest by Theorem~\ref{theorem_simple-series} in Sec.~\ref{sec_optim}.}

\subsubsection{Concentration function $\mathcal{P}(\mathbf{x})$}

\emph{Definition.---}We define %
$\mathcal{P}:\mathbb{R}_{+}^m\to\mathbb{R}_{+}^d$ ($m\ge d$) by the following pseudocode:

\begin{algorithmic}
	\Function{$\mathcal{P}$}{$\mathbf{x}$}: \Comment{$\mathbf{x}^{\downarrow}=\left(x^{\downarrow}_1,x^{\downarrow}_2,\cdots,x^{\downarrow}_m\right)$}
	\State $s\gets\sum_{j=1}^{m}x_j$
	\For{$l\gets 1, 2, \cdots, d$}
	\State $\chi_l\gets \max\left\{x^{\downarrow}_l,s/\left(d+1-l\right)\right\}$
	\State $s\gets s-\chi_l$
	\EndFor
	\State \Return $\boldsymbol{\chi}$ \Comment{$\boldsymbol{\chi}=\left(\chi_1,\cdots,\chi_d\right)$}
	\EndFunction
\end{algorithmic}
\vspace{-8mm}
\begin{equation}
\label{eq_p}
\end{equation}

The above definition implies that $\mathcal{P}(\mathbf{x})$ is always given in descending order, $\mathcal{P}(\mathbf{x})\equiv\left[\mathcal{P}(\mathbf{x})\right]^{\downarrow}$, and
\begin{equation}
\label{eq_p1}
\left[\mathcal{P}(\mathbf{x})\right]_j\ge x^{\downarrow}_j, \quad j=1,2,\cdots,d.
\end{equation}
Moreover, if $\left[\mathcal{P}(\mathbf{x})\right]_l > x^{\downarrow}_l$ for some $l$, then
\begin{equation}
\label{eq_p2}
\left[\mathcal{P}(\mathbf{x})\right]_j=\left[\mathcal{P}(\mathbf{x})\right]_l \text{ and } \left[\mathcal{P}(\mathbf{x})\right]_j > x^{\downarrow}_j 
\end{equation}
for all $l \le j \le d$.

\emph{Properties.---}$\mathcal{P}(\mathbf{x})$ has the following properties.
\begin{itemize}
	\item {Permutation invariance [$\mathbf{x}\to\text{perm}(\mathbf{x})$]:}
	\begin{equation}
	\label{eq_p_perm}
	\mathcal{P}(\mathbf{x})=\mathcal{P}(\text{perm}(\mathbf{x})).
	\end{equation}
	\item {Trace preserving:}
	\begin{equation}
	\label{eq_p_trace}
	\tr(\mathcal{P}(\mathbf{x}))=\tr(\mathbf{x}).
	\end{equation}
	\item {Isotone:}
	\begin{equation}
	\label{eq_p_isotone}
	\mathcal{P}(\mathbf{x})\succ \mathcal{P}(\mathbf{y}) \text{ if } \mathbf{x}\succ \mathbf{y}.
	\end{equation}
	\item {Commutativity:}
	\begin{equation}
	\label{eq_p_commut}
	\mathcal{P}(\mathbf{x}\otimes\mathbf{y})=\mathcal{P}(\mathbf{y}\otimes\mathbf{x}).
	\end{equation}	
	\item  {Associativity:}
	\begin{equation}
	\label{eq_p_assoc}
	\mathcal{P}(\mathcal{P}(\mathbf{x}\otimes\mathbf{y})\otimes\mathbf{z})= \mathcal{P}(\mathbf{x}\otimes \mathcal{P}(\mathbf{y}\otimes\mathbf{z})).
	\end{equation}
\end{itemize}
Equations~\eqref{eq_p_perm},~\eqref{eq_p_trace}~and~\eqref{eq_p_commut} all hold by definition [Eq.~\eqref{eq_p}]. Again, Eq.~\eqref{eq_p_isotone} holds as a result of Theorem~\ref{theorem_convex-perm-isotone}, given the facts that $\mathcal{P}$ is permutation invariant [Eq.~\eqref{eq_p_perm}] and convex (Lemma~\ref{theorem_d_convex}). 

In particular,  Eq.~\eqref{eq_p_assoc} is valid because, on the one hand, by Eq.~\eqref{eq_p_isotone}, we have $\mathcal{P}(\mathcal{P}(\mathbf{x}\otimes\mathbf{y})\otimes\mathbf{z}) \succ \mathcal{P}(\mathbf{x}\otimes\mathbf{y}\otimes\mathbf{z})$; on the other hand, since $\mathcal{P}(\mathbf{x}\otimes\mathbf{y}\otimes\mathbf{z})\oplus(\overbrace{0,0,\cdots,0}^{d^2-d})\succ \mathcal{P}(\mathbf{x}\otimes\mathbf{y})\otimes\mathbf{z}$, by Lemma~\ref{theorem_d_optimal} we have $\mathcal{P}(\mathbf{x}\otimes\mathbf{y}\otimes\mathbf{z}) \succ \mathcal{P}(\mathcal{P}(\mathbf{x}\otimes\mathbf{y})\otimes\mathbf{z})$. Thus we must have $\mathcal{P}(\mathbf{x}\otimes\mathbf{y}\otimes\mathbf{z}) = \mathcal{P}(\mathcal{P}(\mathbf{x}\otimes\mathbf{y})\otimes\mathbf{z})$ since $\mathcal{P}$ is permutation invariant. The same is also true for  $\mathcal{P}(\mathbf{x}\otimes \mathcal{P}(\mathbf{y}\otimes\mathbf{z}))$. Hence together we have $\mathcal{P}(\mathcal{P}(\mathbf{x}\otimes\mathbf{y})\otimes\mathbf{z})= \mathcal{P}(\mathbf{x}\otimes \mathcal{P}(\mathbf{y}\otimes\mathbf{z}))= \mathcal{P}(\mathbf{x}\otimes \mathbf{y}\otimes\mathbf{z})$.

\emph{Implementability by LOCC.---}
The definition [Eq.~\eqref{eq_p1}] immediately implies that
\begin{eqnarray}
\label{eq_p_locc}
\mathcal{P}(\boldsymbol{\lambda})\oplus(\overbrace{0,0,\cdots,0}^{m-d})\succ\boldsymbol{\lambda}.
\end{eqnarray} 
Thus, any new bipartite state of Schmidt numbers $\mathcal{P}(\boldsymbol{\lambda})$ 
can be produced deterministically (i.e.,~$p_{\boldsymbol{\lambda}\to \mathcal{P}(\boldsymbol{\lambda})}=1$) from a bipartite state of Schmidt numbers $\boldsymbol{\lambda}$ by the concentration protocol [Eq.~\eqref{eq_concentrate}]. {However, the full LOCC construction is complicated, and only a sketch was originally given in Ref.~\cite{nielsen_n99}. 
An example of the full construction of $\mathcal{P}(\boldsymbol{\lambda})$ is given in Sec.~\ref{sec_implement}.}

{The choice of $\mathbf{x}$ such that $p_{\boldsymbol{\lambda}\to \mathbf{x}}=1$ is not unique. The motivation behind our special choice of $\mathcal{P}(\boldsymbol{\lambda})$ [Eq.~\eqref{eq_p}] is that it denotes the unique state that contains the maximum entanglement to be concentrated deterministically from $\boldsymbol{\lambda}$ after reducing the dimension. This is made manifest by Theorem~\ref{theorem_simple-parallel} in Sec.~\ref{sec_optim}.}

{\section{Optimality} \label{sec_optim}}
We have shown that the series and parallel rules in Table~\ref{table_qdet-rules} can be implemented by the entanglement swapping and concentration protocols. 
Naturally, we are interested in how good the DET scheme built by these rules is. Here, we proceed to study the {optimality} of these rules. As we will see, the degree of optimality of our series and parallel rules actually varies for different series-parallel network topologies. 

When discussing the optimality of the series and parallel rules, we will compare our specific rules with generic probabilistic protocols. We will show that, to some degree, the deterministic outcomes by the series and parallel rules are still the best, even when we \emph{relax the requirement of determinacy} and consider the average of general probabilistic ensembles of outcomes. 
Since an entanglement-efficient protocol that yields generic probabilistic outcomes would necessarily require %
more storage and more sophisticated algorithms to keep track of every possible outcome, the complexity of the full scheme would inevitably scale with the QN size, suggesting that a DET scheme is more practical for entanglement transmission over a large-scale QN.

\emph{Concurrence monotones.---}It is necessary to introduce a proper {entanglement measure} on the Schmidt numbers $\boldsymbol{\lambda}$. Although $\boldsymbol{\lambda}$ only admits a preorder, the entanglement measure as a real function admits a {strict total order}, which can be used to quantify the amount of entanglement. We will adopt a special family of entanglement monotones, called the {concurrence monotones}~\cite{k-concurrence_g05}, $C_k$, $k=1,2,\cdots,d$ (see Appendix~\ref{sec_concurrence} for more details). 
We will show that the series and parallel rules optimize the task of entanglement transmission for series-parallel QNs when measured by the average or worst-case $k$-concurrence, or $G$-concurrence (another name for $C_d$), depending on the network topology.

{\subsection{Some lemmas}}
We first prove some lemmas for the swapping and concentration functions, $\mathcal{S}(\mathbf{x},\mathbf{y})$ %
and $\mathcal{P}(\mathbf{x})$. %
\begin{lemma}
	\label{theorem_s_convex}
	(Convexity)	
	$\mathcal{S}(\mathbf{x}_\alpha,\mathbf{z})$ is convex, i.e.,
	\begin{equation*}
	\label{theorem_s_convex_1}
	\sum_\alpha p_\alpha \mathcal{S}(\mathbf{x}_\alpha,\mathbf{z}) \succ \mathcal{S}(\sum_\alpha p_\alpha \mathbf{x}_\alpha,\mathbf{z}), \quad p_\alpha\in\mathbb{R}_{+}, \quad \mathbf{z}\in\mathbb{R}_{+}^{d}.
	\end{equation*}
\end{lemma}
\begin{proof}
	Equation~\eqref{eq_s} can be rewritten as
	\begin{equation}
	\label{theorem_s_convex_2}
	\mathcal{S}(\mathbf{x},\mathbf{y})=d\times\mathrm{eig}(\text{diag}(\mathbf{x}^{\downarrow})\mathbf{V}^\dagger\text{diag}(\mathbf{y}^{\downarrow})\mathbf{V}),
	\end{equation}
	where ``$\mathrm{eig}$'' denotes the eigenvalues, arranged in descending order. Thus
	\begin{eqnarray*}
		\label{theorem_s_convex_3}
		\mathcal{S}(\sum_\alpha p_\alpha \mathbf{x}_\alpha,\mathbf{z})%
		&\prec& d \sum_\alpha p_\alpha\mathrm{eig}(  \text{diag}(\mathbf{x}^{\downarrow}_\alpha)\mathbf{V}^\dagger\text{diag}(\mathbf{z}^{\downarrow})\mathbf{V})\nonumber\\
		&=&\sum_\alpha p_\alpha \mathcal{S}(\mathbf{x}_\alpha,\mathbf{z})
	\end{eqnarray*}
	by Lidskii's theorem on eigenvalues (Theorem III.4.1 of Ref.~\cite{matrix-anal}). %
\end{proof}

\begin{lemma}
	\label{theorem_s_det}
	(Determinant preserving)
	\begin{equation*}
	\label{theorem_s_det_1}
	\det(\mathcal{S}(\mathbf{x},\mathbf{y}))=d^d \det(\mathbf{x})\det(\mathbf{y}),
	\end{equation*}
	where $\det(\mathbf{x})\equiv \det(\mathrm{diag}(\mathbf{x}))=x_1x_2\cdots x_d$.
\end{lemma}
\begin{proof}
	By Eq.~\eqref{theorem_s_convex_2}, we directly have  
	$\det(\mathcal{S}(\mathbf{x},\mathbf{y}))=d^d
	\det(\mathbf{x})\left|\det(\mathbf{V})\right|^2\det(\mathbf{y})=d^d \det(\mathbf{x})\det(\mathbf{y})$.
\end{proof}

\begin{lemma}
	\label{theorem_s_adj}
	(Duality)
	Let $\mathbf{z}=\mathcal{S}(\mathbf{x},\mathbf{y})$. Then,
	\begin{equation*}
	\label{theorem_s_adj_1}
	\mathrm{adj}(\mathbf{z})^{\downarrow}=d^{d-2}\mathcal{S}(\mathrm{adj}(\mathbf{x}),\mathrm{adj}(\mathbf{y})),
	\end{equation*}
	where $\mathrm{adj}(\mathbf{x})\equiv %
	\left(\det(\mathbf{x})/x_1,\det(\mathbf{x})/x_2,\cdots,\det(\mathbf{x})/x_d\right)$.
\end{lemma}
\begin{proof}
	By Eq.~\eqref{theorem_s_convex_2},
	\begin{eqnarray*}
		\label{theorem_s_adj_3}
		&& \mathcal{S}(\mathrm{adj}(\mathbf{x}),\mathrm{adj}(\mathbf{y}))\nonumber\\
		&=& d\times\mathrm{eig}(\text{diag}(\mathrm{adj}(\mathbf{x}^{\downarrow}))\mathbf{V}^\dagger\text{diag}(\mathrm{adj}(\mathbf{y}^{\downarrow}))\mathbf{V})\nonumber\\
		&=& d\times\mathrm{eig}(\mathrm{adj}(\text{diag}(\mathbf{x}^{\downarrow}))\mathrm{adj}(\mathbf{V})\mathrm{adj}(\text{diag}(\mathbf{y}^{\downarrow}))\mathrm{adj}(\mathbf{V}^\dagger))\nonumber\\
		&=& d\times\mathrm{eig}(\mathrm{adj}(\mathbf{V}^\dagger\text{diag}(\mathbf{y}^{\downarrow})\mathbf{V}\text{diag}(\mathbf{x}^{\downarrow})))\nonumber\\
		&=& d\times\mathrm{adj}(\mathrm{eig}(\mathbf{V}^\dagger\text{diag}(\mathbf{y}^{\downarrow})\mathbf{V}\text{diag}(\mathbf{x}^{\downarrow})))^{\downarrow}\nonumber\\
		&=& d\times\mathrm{adj}(\mathcal{S}(\mathbf{x},\mathbf{y})/d)^{\downarrow}=d^{2-d}\mathrm{adj}(\mathbf{z})^{\downarrow},
	\end{eqnarray*}
	where we have used several properties of the adjugate of a matrix $\boldsymbol{\Psi}$, e.g.,~$\mathrm{adj}(\boldsymbol{\Psi})=\det(\boldsymbol{\Psi}) \boldsymbol{\Psi}^{-1}$ and $\mathrm{adj}(c \boldsymbol{\Psi})=c^{d-1}\mathrm{adj}(\boldsymbol{\Psi})$ for $c\in\mathbb{R}$.
\end{proof}

\begin{lemma}
	\label{theorem_d_optimal}
	(Majorization extremity)
	Let $\mathbf{z}'\in\mathbb{R}_{+}^{d}$. If $\mathbf{z}'\oplus(\overbrace{0,0,\cdots,0}^{m-d})\succ\mathbf{x}$, then $\mathbf{z}'\succ \mathcal{P}(\mathbf{x})$.
\end{lemma}
\begin{proof}
	We will prove this lemma by induction. Suppose there exists $\left[\mathcal{P}(\mathbf{x})\right]_l>x^{\downarrow}_l$ for some $l$, and
	\begin{equation}
	\label{theorem_d_optimal_1}
	\sum_{j=1}^{k}z'^{\downarrow}_j \ge \sum_{j=1}^{k}\left[\mathcal{P}(\mathbf{x})\right]_j, \quad \forall k=1,\cdots,l-1;
	\end{equation}
	we would like to prove
	\begin{equation}
	\label{theorem_d_optimal_2}
	\sum_{j=1}^{l}z'^{\downarrow}_j \ge \sum_{j=1}^{l}\left[\mathcal{P}(\mathbf{x})\right]_j.
	\end{equation}
	Indeed,
	\begin{eqnarray}
	\label{theorem_d_optimal_3}
	&&\sum_{j=1}^{l}z'^{\downarrow}_j =\sum_{j=1}^{l-1}z'^{\downarrow}_j + z'^{\downarrow}_l \ge
	\sum_{j=1}^{l-1}z'^{\downarrow}_j + \frac{\sum_{j=1}^{d}z'^{\downarrow}_j-\sum_{j=1}^{l-1}z'^{\downarrow}_j}{d+1-l}\nonumber\\
	&&\ge\sum_{j=1}^{l-1}\left[\mathcal{P}(\mathbf{x})\right]_j+ \frac{\sum_{j=1}^{d}\left[\mathcal{P}(\mathbf{x})\right]_j-\sum_{j=1}^{l-1}\left[\mathcal{P}(\mathbf{x})\right]_j}{d+1-l}\nonumber\\
	&&=\sum_{j=1}^{l-1}\left[\mathcal{P}(\mathbf{x})\right]_j+ \left(d+1-l\right)^{-1}{\sum_{j=l}^{d}\left[\mathcal{P}(\mathbf{x})\right]_j}\nonumber\\
	&&=\sum_{j=1}^{l-1}\left[\mathcal{P}(\mathbf{x})\right]_j+\left[\mathcal{P}(\mathbf{x})\right]_l=\sum_{j=1}^{l}\left[\mathcal{P}(\mathbf{x})\right]_j,
	\end{eqnarray}
	where the first inequality holds because the maximum is never less than the mean and the second inequality holds because %
	$\tr (\mathcal{P}(\mathbf{x}))=\tr (\mathbf{z}')$
	and $1-\left(d+1-l\right)^{-1}\ge0$. Hence Eq.~\eqref{theorem_d_optimal_2} is proved.
	
	Finally, choose the minimum $l$ that satisfies $\left[\mathcal{P}(\mathbf{x})\right]_l>x^{\downarrow}_l$. In other words, $\left[\mathcal{P}(\mathbf{x})\right]_j=x^{\downarrow}_j$ for $j <l$. Hence
	\begin{eqnarray*}
		\label{theorem_d_optimal_4}
		&&\sum_{j=1}^{k}\left[\mathcal{P}(\mathbf{x})\right]_j=\sum_{j=1}^{k}x^{\downarrow}_j \le \sum_{j=1}^{k}  z'^{\downarrow}_j, \quad \forall k=1,\cdots,l-1, \qquad
	\end{eqnarray*}
	which completes the induction.	
\end{proof}

\begin{lemma}
	\label{theorem_d_convex}
	(Convexity)	
	$\mathcal{P}(\mathbf{x}_\alpha)$ is convex, i.e.,
	\begin{equation*}
	\label{eq_p-convex}
	\sum_\alpha p_\alpha \mathcal{P}(\mathbf{x}_\alpha) \succ \mathcal{P}(\sum_\alpha p_\alpha \mathbf{x}_\alpha), \quad p_\alpha\in\mathbb{R}_{+}.
	\end{equation*}
\end{lemma}
\begin{proof}
	Similar to Lemma~\ref{theorem_d_optimal}, we will prove this lemma by induction. Suppose there exists $\left[\mathcal{P}(\sum_\alpha p_\alpha \mathbf{x}_\alpha)\right]_l>\left[\sum_\alpha p_\alpha \mathbf{x}_\alpha\right]^{\downarrow}_l$ for some $l$, and
	\begin{eqnarray}
	\label{theorem_d_convex_1}
	&&\sum_{j=1}^{k}\left[\sum_\alpha p_\alpha \mathcal{P}(\mathbf{x}_\alpha)\right]^{\downarrow}_j \ge \sum_{j=1}^{k}\left[\mathcal{P}(\sum_\alpha p_\alpha \mathbf{x}_\alpha)\right]_j,
	\end{eqnarray}
	$\forall k=1,\cdots,l-1$; we can prove
	\begin{equation}
	\label{theorem_d_convex_2}
	\sum_{j=1}^{l}\left[\sum_\alpha p_\alpha \mathcal{P}(\mathbf{x}_\alpha)\right]^{\downarrow}_j \ge \sum_{j=1}^{l}\left[\mathcal{P}(\sum_\alpha p_\alpha \mathbf{x}_\alpha)\right]_j
	\end{equation}
    the same way as in Eq.~\eqref{theorem_d_optimal_3}.
	Now, choose the minimum $l$ that satisfies $\left[\mathcal{P}(\sum_\alpha p_\alpha \mathbf{x}_\alpha)\right]_l>\left[\sum_\alpha p_\alpha \mathbf{x}_\alpha\right]^{\downarrow}_l$. In other words, $\left[\mathcal{P}(\sum_\alpha p_\alpha \mathbf{x}_\alpha)\right]_j=\left[\sum_\alpha p_\alpha \mathbf{x}_\alpha\right]^{\downarrow}_j$ for $j <l$. Hence
	\begin{eqnarray*}
		\label{theorem_d_convex_4}
		&&\sum_{j=1}^{k}\left[\mathcal{P}(\sum_\alpha p_\alpha \mathbf{x}_\alpha)\right]_j =
		\sum_{j=1}^{k}\left[\sum_\alpha p_\alpha \mathbf{x}_\alpha\right]^{\downarrow}_j\nonumber\\
		&\le& \sum_{j=1}^{k}\left[\sum_\alpha p_\alpha \mathbf{x}_\alpha^{\downarrow}\right]^{\downarrow}_j
		\le \sum_{j=1}^{k}\left[\sum_\alpha p_\alpha \mathcal{P}(\mathbf{x}_\alpha)\right]^{\downarrow}_j,
	\end{eqnarray*}
	$\forall k=1,\cdots,l-1$, which completes the induction. 
\end{proof}

\begin{lemma}
	\label{theorem_d_log}
	(A sum-product mixing majorization inequality)
	Let $\mathbf{x},\mathbf{y},\mathbf{z}\in\mathbb{R}_{+}^{m}$. If %
	\begin{equation*}
	\ln \mathbf{x}^{\downarrow}+ \ln \mathbf{y}^{\downarrow} \succ_w \ln \mathbf{z}^{\downarrow},
	\end{equation*}
	then
	\begin{equation*}
	\label{theorem_d_log_0}
	\ln \mathcal{P}(\mathbf{x})+\ln \mathcal{P}(\mathbf{y})\succ_w\ln \mathcal{P}(\mathbf{z}).
	\end{equation*}
\end{lemma}
\begin{proof}
	To begin with,  w.l.o.g.~we will just assume $m=d+1$, which means that $\mathcal{P}(\mathbf{x})$ is only one dimension less than $\mathbf{x}$. Indeed, for any $m>d$, the concentration process [Eq.~\eqref{eq_p}] can be equally constructed by applying such a one-dimension-less $\mathcal{P}(\mathbf{x})$ to $\mathbf{x}$ recursively $m-d$ times. Thus, if our statement holds for $m=d+1$, then it holds for general $m$ too.
	
	First, we would like to prove the following inequality: %
	\begin{equation}
	\label{theorem_d_log_1}
	\left(\prod_{j=1}^{d}\left[\mathcal{P}(\mathbf{x})\right]_j\right)\left(\prod_{j=1}^{d}\left[\mathcal{P}(\mathbf{y})\right]_j\right) \ge \prod_{j=1}^{d}\left[\mathcal{P}(\mathbf{z})\right]_j.
	\end{equation}
	By the definition of $\mathcal{P}(\mathbf{x})$ [Eq.~\eqref{eq_p2}], we have
	\begin{equation*}
	\label{theorem_d_log_2}
	\prod_{j=1}^{d}\left[\mathcal{P}(\mathbf{z})\right]_j=\left(\prod_{j=1}^{l-1}z^{\downarrow}_j\right)\left[{\sum_{j=l}^{d+1} z^{\downarrow}_j}/{\left(d+1-l\right)}\right]^{d+1-l},
	\end{equation*}
	where $l$ separates $\mathcal{P}(\mathbf{z})$ and satisfies $\left[\mathcal{P}(\mathbf{z})\right]_j=z^{\downarrow}_j$ for $j <l$ and $\left[\mathcal{P}(\mathbf{z})\right]_j>z^{\downarrow}_j$ for $l \le j \le d$. 	

	Now we take advantage of a special reverse AM--GM inequality which we will introduce and prove in Appendix~\ref{sec_reverse}. Let $\Delta=z^{\downarrow}_{d+1}$, $\varepsilon_{j}=z^{\downarrow}_{d+1-j}$, and $E_k=\Delta+\sum_{j=1}^{k}\varepsilon_{j}$, $k=1,2,\cdots,d+1-l$. We can see that $\varepsilon_{k} \le \varepsilon_{k+1}\le {E_k}/k$, $\forall k$ [implied by Eq.~\eqref{eq_p}]. Thus, using Corollary~\ref{theorem_reverse-d},
	\begin{eqnarray*}
		\label{theorem_d_log_3}
		\left[{\sum_{j=l}^{d+1} z^{\downarrow}_j}/{\left(d+1-l\right)}\right]^{d+1-l}&&\le \exp({z^{\downarrow}_{d+1}/z^{\downarrow}_d}) \prod_{j=l}^{d} z^{\downarrow}_j\nonumber\\
		&&\le \left(1+\sqrt{{z^{\downarrow}_{d+1}}/{z^{\downarrow}_d}}\right)^2 \prod_{j=l}^{d} z^{\downarrow}_j,
		\quad\quad
	\end{eqnarray*}
	where the second inequality holds since ${z^{\downarrow}_{d+1}\le z^{\downarrow}_d}$. 
	Thus,
	\begin{eqnarray*}
		\label{theorem_d_log_5}
		&&\prod_{j=1}^{d}\left[\mathcal{P}(\mathbf{z})\right]_j \le 
		\left(1+\sqrt{{z^{\downarrow}_{d+1}}/{z^{\downarrow}_d}}\right)^2 \prod_{j=1}^{d} z^{\downarrow}_j \nonumber\\
		&=&\left(z^{\downarrow}_{d}+z^{\downarrow}_{d+1}+2\sqrt{z^{\downarrow}_{d} z^{\downarrow}_{d+1}}\right) \prod_{j=1}^{d-1} z^{\downarrow}_j\nonumber\\	
		&\le&  \left(x^{\downarrow}_{d}y^{\downarrow}_{d}+x^{\downarrow}_{d+1}y^{\downarrow}_{d+1}+2\sqrt{x^{\downarrow}_{d}y^{\downarrow}_{d}x^{\downarrow}_{d+1}y^{\downarrow}_{d+1}}\right) \prod_{j=1}^{d-1} x^{\downarrow}_j y^{\downarrow}_j\nonumber\\
		&\le& \left(x^{\downarrow}_{d}+x^{\downarrow}_{d+1}\right)\left(y^{\downarrow}_{d}+y^{\downarrow}_{d+1}\right) \prod_{j=1}^{d-1} x^{\downarrow}_j y^{\downarrow}_j \nonumber\\
		&\le& \left(\prod_{j=1}^{d}\left[\mathcal{P}(\mathbf{x})\right]_j\right)\left(\prod_{j=1}^{d}\left[\mathcal{P}(\mathbf{y})\right]_j\right),
	\end{eqnarray*}	
	where in the third step we use Theorem~\ref{theorem_sum-product}, in the fourth step we use the usual AM--GM inequality,  and in the last step we use the definition of the concentration process [Eq.~\eqref{eq_p}]. Hence Eq.~\eqref{theorem_d_log_1} is proved.
	
	Now we will complete the proof by induction. Suppose there exists $\left[\mathcal{P}(\mathbf{z})\right]_l>z^{\downarrow}_l$ for some $l$, and
	\begin{equation}
	\label{theorem_d_log_6}
	\sum_{j=1}^{k}\ln\left[ \mathcal{P}(\mathbf{x})\right]_j+\sum_{j=1}^{k}\ln\left[ \mathcal{P}(\mathbf{y})\right]_j\ge \sum_{j=1}^{k}\ln \left[\mathcal{P}(\mathbf{z})\right]_j,
	\end{equation}
	$\forall k=1,\cdots,l-1$, we would like to prove
	\begin{equation}
	\label{theorem_d_log_7}
	\sum_{j=1}^{l}\ln \left[\mathcal{P}(\mathbf{x})\right]_j+\sum_{j=1}^{l}\ln\left[ \mathcal{P}(\mathbf{y})\right]_j \ge \sum_{j=1}^{l}\ln \left[ \mathcal{P}(\mathbf{z})\right]_j.
	\end{equation}
	The proof structure is the same as in Eq.~\eqref{theorem_d_optimal_3}, while the only difference is that, for the second inequality in Eq.~\eqref{theorem_d_optimal_3} to hold, we need to show that 
	$\tr (\ln \mathcal{P}(\mathbf{x}))+\tr (\ln \mathcal{P}(\mathbf{y})) \ge \tr (\ln  \mathcal{P}(\mathbf{z}))$.
	This has been proved by Eq.~\eqref{theorem_d_log_1}. Hence Eq.~\eqref{theorem_d_log_7} is proved.
	
	Finally, choose the minimum $l$ that satisfies $\left[\mathcal{P}(\mathbf{z})\right]_l>z^{\downarrow}_l$. In other words, $\left[\mathcal{P}(\mathbf{z})\right]_j=z^{\downarrow}_j$ for $j <l$. Hence
	\begin{eqnarray*}
		\label{theorem_d_log_8}
		&&\sum_{j=1}^{k}\ln \left[\mathcal{P}(\mathbf{z})\right]_j=\sum_{j=1}^{k} \ln z^{\downarrow}_j \le \sum_{j=1}^{k} \ln x^{\downarrow}_j+ \sum_{j=1}^{k} \ln y^{\downarrow}_j\nonumber\\ 
		&&\le \sum_{j=1}^{k}\ln \left[\mathcal{P}(\mathbf{x})\right]_j+\sum_{j=1}^{k}\ln  \left[\mathcal{P}(\mathbf{y})\right]_j, \forall k=1,\cdots,l-1 \qquad
	\end{eqnarray*}
	which completes the induction.
\end{proof}

\subsection{Simple Series}
For simple series network topology [Fig.~\ref{fig_seri}], the swapping protocol is responsible for entanglement transmission. %
What is the maximum of the final average $k$-concurrence [from Eq.~\eqref{eq_swap}]
\begin{equation*}
\sum_\alpha p_\alpha C^{\alpha}_k=\sum_\alpha \left\|\boldsymbol{\Psi}^{\alpha}\right\|_2^2 C_k( \sigma^2({\left\|\boldsymbol{\Psi}^{\alpha}\right\|_2^{-1}}{\boldsymbol{\Psi}^{\alpha}}))
\end{equation*}
between A and B that can be obtained by the swapping protocol? How does it compare with the deterministic series rule? This is answered by the following theorem.
\begin{theorem}
	\label{theorem_simple-series}
	Given a simple series QN, 
	compared with generic entanglement swapping protocols of probabilistic outcomes, 
	the series rule (Table~\ref{table_qdet-rules}) produces the optimal average $G$-concurrence between A and B.
\end{theorem}
\begin{proof}
	It suffices to investigate two links $\boldsymbol{\lambda}_a$ and $\boldsymbol{\lambda}_b$ in series. The proof can be easily generalized to $n$ links.

	Let $\{\mathbf{X}^{\alpha}\}$ denote a set of quantum measurements as used in the entanglement swapping protocol [Eq.~\eqref{eq_swap}]. %
	In general, given probability $p_\alpha$ and the corresponding $k$-concurrence $C^{\alpha}_k$ of outcome $\alpha$, the maximum average $k$-concurrence $\max_{\{\mathbf{X}^{\alpha}\}}\sum_\alpha p_\alpha C^{\alpha}_k$ is intractable, and the corresponding optimal $\{\mathbf{X}^{\alpha}\}$ should implicitly depend on $\boldsymbol{\lambda}_a$ and $\boldsymbol{\lambda}_b$. However, this is not the case for $G$-concurrence. To see this, 
	let $k=d$; then, 
	\begin{eqnarray}
	\label{theorem_simple-series_2}
	&&\sum_\alpha p_\alpha C^{\alpha}_d= d \sum_\alpha \left|\det(\boldsymbol{\Psi}^{\alpha})\right|^{{2}/{d}} \nonumber\\
	=&& d \sum_\alpha \det(\boldsymbol{\lambda}_a)^{1/d} \left|\det(\mathbf{X}^{\alpha})\right|^{2/d} \det(\boldsymbol{\lambda}_b)^{1/d} \nonumber\\
	=&& d^{-1}C_d(\boldsymbol{\lambda}_a) C_d(\boldsymbol{\lambda}_b)\sum_\alpha \left|\det(\mathbf{X}^{\alpha})\right|^{2/d}\nonumber\\
	\le&& d^{-2}C_d(\boldsymbol{\lambda}_a) C_d(\boldsymbol{\lambda}_b)\sum_\alpha \tr(\mathbf{X}^{\alpha\dagger}\mathbf{X}^{\alpha})=C_d(\boldsymbol{\lambda}_a) C_d(\boldsymbol{\lambda}_b),\qquad
	\end{eqnarray}
	where the AM--GM inequality is used in the second-to-last step and Eq.~\eqref{eq_povm} is used in the last step. Equation~\eqref{theorem_simple-series_2} indicates that the final average $G$-concurrence will never be greater than the product of the $G$-concurrences of $\boldsymbol{\lambda}_a$ and $\boldsymbol{\lambda}_b$. 
 
 {Note that this proof was originally given by Gour in Ref.~\cite{k-concurrence_g05}. The validity of the proof comes from the unique feature of the multiplicity of determinants, which is not applicable to other $k$-concurrence monotones~\cite{k-concurrence_g05}. A summary of the use of determinant-based entanglement measures compared with other probabilistic measures for certain QN topologies can be found in Ref.~\cite{QEP-detail_pcalw08}.}
	
	It remains to be shown that the equality in Eq.~\eqref{theorem_simple-series_2} holds for $\mathcal{S}(\boldsymbol{\lambda}_a,\boldsymbol{\lambda}_b)$. By Lemma~\ref{theorem_s_det},
	\begin{eqnarray}
	\label{theorem_simple-series_3}
	&& C_d(\mathcal{S}(\boldsymbol{\lambda}_a,\boldsymbol{\lambda}_b))= d \det(\mathcal{S}(\boldsymbol{\lambda}_a,\boldsymbol{\lambda}_b))^{1/d}\nonumber\\
	&=& d^2 \det(\boldsymbol{\lambda}_a)^{1/d} \det(\boldsymbol{\lambda}_b)^{1/d}= C_d(\boldsymbol{\lambda}_a) C_d(\boldsymbol{\lambda}_b).
	\end{eqnarray}
	Thus,
	\begin{equation}
	\label{theorem_simple-series_4}
	\sum_\alpha p_\alpha C_d^\alpha \le C_d(\mathcal{S}(\boldsymbol{\lambda}_a,\boldsymbol{\lambda}_b)),
	\end{equation}
	which completes the proof. 
\end{proof}

\emph{Remark.---}Note that Eq.~\eqref{theorem_simple-series_4} says nothing about the majorization preorder. In fact, usually there is
\begin{equation}
\label{theorem_simple-series_5}
\sum_\alpha p_\alpha \sigma^2(\boldsymbol{\Psi}^{\alpha}/\left\|\boldsymbol{\Psi}^{\alpha}\right\|_2)^{\downarrow} = \sum_\alpha \sigma^2(\boldsymbol{\Psi}^{\alpha})^{\downarrow}\nsucc \mathcal{S}(\boldsymbol{\lambda}_a,\boldsymbol{\lambda}_b).
\end{equation}
Would the inequality in Eq.~\eqref{theorem_simple-series_5} hold, one could prove Eq.~\eqref{theorem_simple-series_4} for not only $k=d$ but $k<d$ as well. Unfortunately, this is not true.

\subsection{Simple Parallel}
For simple parallel network topology [Fig.~\ref{fig_para}], the concentration protocol is responsible for entanglement transmission. Similarly, we ask:
What is the maximum of the final average $k$-concurrence $\sum_\alpha p_\alpha C^{\alpha}_k$ between A and B that can be obtained by the concentration protocol? How does it compare with the deterministic parallel rule? The answer is given by the following theorem.
\begin{theorem}
	\label{theorem_simple-parallel}
	Given a simple parallel QN, 
	compared with generic entanglement concentration protocols of probabilistic outcomes, 
	the parallel rule (Table~\ref{table_qdet-rules}) produces the optimal average $k$-concurrence between A and B for $k=1,2,\cdots,d$.
\end{theorem}
\begin{proof}
	Again, it suffices to look at two links $\boldsymbol{\lambda}_a$ and $\boldsymbol{\lambda}_b$ in parallel. The proof can be easily generalized to $n$ links.
	
	If we can obtain an ensemble of outcomes $\boldsymbol{\lambda}_\alpha'$, each with probability $p_\alpha$, from $\boldsymbol{\lambda}_a\otimes\boldsymbol{\lambda}_b$, then, by the fundamental limit of LOCC~\cite{monotone_v00},
	\begin{equation*}
	\label{theorem_simple-parallel_1}
	\left(\sum_\alpha p_\alpha \boldsymbol{\lambda}_\alpha'^{\downarrow}\oplus(\overbrace{0,0,\cdots,0}^{d^2-d}) \right)\succ \boldsymbol{\lambda}_a\otimes\boldsymbol{\lambda}_b.
	\end{equation*}
	Hence, by Lemma~\ref{theorem_d_optimal},
	\begin{eqnarray}
	\label{theorem_simple-parallel_2}
	\sum_\alpha p_\alpha \boldsymbol{\lambda}_\alpha'^{\downarrow} \succ \mathcal{P}(\boldsymbol{\lambda}_a\otimes\boldsymbol{\lambda}_b),
	\end{eqnarray}
	and thus
	\begin{eqnarray*}
		\label{theorem_simple-parallel_3}
		\sum_\alpha p_\alpha C_k(\boldsymbol{\lambda}_\alpha') \le C_k(\sum_\alpha p_\alpha \boldsymbol{\lambda}_\alpha'^{\downarrow}) \le C_k(\mathcal{P}(\boldsymbol{\lambda}_a\otimes\boldsymbol{\lambda}_b))\quad
	\end{eqnarray*}
	for $k=1,2,\cdots,d$.
\end{proof}

\emph{Remark.---} Note that unlike Eq.~\eqref{theorem_simple-series_5} for the series rule, Eq.~\eqref{theorem_simple-parallel_2} does hold. Therefore the optimality of the parallel rule is for all $k$-concurrence, not just $G$-concurrence.
{This implies that $\mathcal{P}(\boldsymbol{\lambda}_a\otimes\boldsymbol{\lambda}_b)$ always contains the maximum allowable entanglement that can be concentrated from $\boldsymbol{\lambda}_a\otimes\boldsymbol{\lambda}_b$ on average.}

\subsection{Parallel-then-Series}
For parallel-then-series network topology [Fig.~\ref{fig_para-then-seri}], both the swapping and concentration protocols are responsible for entanglement transmission. We prove the following theorem.
\begin{theorem}
	\label{theorem_parallel-then-series}
	Given a parallel-then-series QN, 
	compared with generic entanglement swapping and concentration protocols of probabilistic outcomes, 
	the series and parallel rules (Table~\ref{table_qdet-rules}) produce the optimal average $G$-concurrence between A and B.
\end{theorem}
\begin{proof}
	It suffices to prove the theorem for four links: $\boldsymbol{\lambda}_a$ and $\boldsymbol{\lambda}_b$ are connected in parallel, $\boldsymbol{\lambda}_c$ and $\boldsymbol{\lambda}_d$ are connected in parallel, and then the two parallel groups are connected in series.
	Now, let $\boldsymbol{\Psi}^{\alpha}=\text{diag}(\boldsymbol{\lambda}_a\otimes \boldsymbol{\lambda}_b)^{1/2} \mathbf{X}^{\alpha} \text{diag}(\boldsymbol{\lambda}_c\otimes \boldsymbol{\lambda}_d)^{1/2}$, as used in the entanglement swapping protocol [Eq.~\eqref{eq_swap}]. The difference, however, is that $\boldsymbol{\Psi}^{\alpha}$ is not a $d\times d$ matrix, but a $d^2\times d^2$ matrix. This is because we are considering a generic swapping protocol, which must be applied to the full $d^2$ dimensions. Then, the final $d$-dimensional outcome $\alpha$ will be given by $\mathcal{P}(\sigma^2(\boldsymbol{\Psi}^{\alpha}))$.
	
	By the Gelfand-Naimark theorem on singular values (Theorem III.4.5 of Ref.~\cite{matrix-anal}), we have
	\begin{equation}
	\label{theorem_parallel-then-series_1}
	\ln \sigma^2(\boldsymbol{\Psi}^{\alpha}) \prec \ln \left(\boldsymbol{\lambda}_a\otimes \boldsymbol{\lambda}_b\right)^{\downarrow}+\ln \sigma^2(\mathbf{X}^{\alpha})+\ln \left(\boldsymbol{\lambda}_c\otimes \boldsymbol{\lambda}_d\right)^{\downarrow}.
	\end{equation}
	Thus, by Lemma~\ref{theorem_d_log} we have, in particular,
	\begin{eqnarray*}
		&&\det(\mathcal{P}(\sigma^2(\boldsymbol{\Psi}^{\alpha})))\nonumber\\
		&\le& 
		\det(\mathcal{P}(\boldsymbol{\lambda}_a\otimes \boldsymbol{\lambda}_b))
		\det(\mathcal{P}(\sigma^2(\mathbf{X}^{\alpha})))
		\det(\mathcal{P}(\boldsymbol{\lambda}_c\otimes \boldsymbol{\lambda}_d)).%
	\end{eqnarray*}
	Hence, similar to Eq.~\eqref{theorem_simple-series_2},
	\begin{eqnarray}
	\label{theorem_parallel-then-series_3}
	&&\sum_\alpha p_\alpha C^{\alpha}_d= d \sum_\alpha \det(\mathcal{P}(\sigma^2(\boldsymbol{\Psi}^{\alpha})))^{{1}/{d}} \nonumber\\
	\le&& d^{-2}C_d(\mathcal{P}(\boldsymbol{\lambda}_a\otimes \boldsymbol{\lambda}_b)) C_d(\mathcal{P}(\boldsymbol{\lambda}_c\otimes \boldsymbol{\lambda}_d))\sum_\alpha \tr(\mathcal{P}(\sigma^2(\mathbf{X}^{\alpha})))\nonumber\\
	=&& d^{-2}C_d(\mathcal{P}(\boldsymbol{\lambda}_a\otimes \boldsymbol{\lambda}_b)) C_d(\mathcal{P}(\boldsymbol{\lambda}_c\otimes \boldsymbol{\lambda}_d))\sum_\alpha \tr(\mathbf{X}^{\alpha\dagger}\mathbf{X}^{\alpha})\nonumber\\
	=&& C_d(\mathcal{P}(\boldsymbol{\lambda}_a\otimes \boldsymbol{\lambda}_b)) C_d(\mathcal{P}(\boldsymbol{\lambda}_c\otimes \boldsymbol{\lambda}_d)),\qquad
	\end{eqnarray}
	where the AM--GM inequality is used in the second step, Eq.~\eqref{eq_p_trace} is used in the third step, and Eq.~\eqref{eq_povm} is used in the last step. Again, Eq.~\eqref{theorem_parallel-then-series_3} indicates that the final average $G$-concurrence will never be greater than the product of the $G$-concurrences of $\mathcal{P}(\boldsymbol{\lambda}_a\otimes \boldsymbol{\lambda}_b)$ and $\mathcal{P}(\boldsymbol{\lambda}_c\otimes \boldsymbol{\lambda}_d)$, which is equal to $C_d(\mathcal{S} (\mathcal{P}(\boldsymbol{\lambda}_a\otimes \boldsymbol{\lambda}_b), \mathcal{P}(\boldsymbol{\lambda}_c\otimes \boldsymbol{\lambda}_d)))$ [Eq.~\eqref{theorem_simple-series_3}].
\end{proof}

\emph{Remark.---}At first glimpse, it seems that we would have more degrees of freedom to design a swapping protocol directly on the full $d^2$-dimensional $\boldsymbol{\lambda}_a\otimes \boldsymbol{\lambda}_b$ and $\boldsymbol{\lambda}_c\otimes \boldsymbol{\lambda}_d$, and we would get a better result in terms of $G$-concurrence. Interestingly though, Theorem~\ref{theorem_parallel-then-series} says that the best strategy is to first reduce them to $d$-dimensional $\mathcal{P}(\boldsymbol{\lambda}_a\otimes \boldsymbol{\lambda}_b)$ and $\mathcal{P}(\boldsymbol{\lambda}_c\otimes \boldsymbol{\lambda}_d)$ and then apply a swapping protocol on them. Our result implies a strong practical convenience in entanglement transmission and that it is in fact unnecessary to store all $d^m$ dimensions for $m$ parallel links. We can safely concentrate and produce a $d$-dimensional link from the $m$ parallel links and then swap it with other links connected in series. The final $G$-concurrence is guaranteed to be optimal.

Note particularly that the nested repeater protocol~\cite{q-repeater_bdcz98} is designed to swap $\boldsymbol{\lambda}_a\otimes \boldsymbol{\lambda}_b$ and $\boldsymbol{\lambda}_c\otimes \boldsymbol{\lambda}_d$ using two independent swapping protocols on $a,c$ and $b,d$, respectively, before applying concentration protocols on the outcomes. However, the tensor product of the two swapping protocols is just a special case of a $d^2$-dimensional swapping protocol. Thus the nested concentration protocol is neither the optimal nor the convenient approach to produce the final average $G$-concurrence.

\subsection{Series-then-Parallel}
For series-then-parallel network topology [Fig.~\ref{fig_seri-then-para}], the final average $k$-concurrence between A and B is \emph{not} maximized by the series and parallel rules but instead by some generic nondeterministic swapping and concentration protocols. 

A counterexample can be easily constructed. Let $d=2$ and $\boldsymbol{\lambda}_a=\boldsymbol{\lambda}_b=\boldsymbol{\lambda}_c=\left(0.9,0.1\right)$. $\boldsymbol{\lambda}_a$ and $\boldsymbol{\lambda}_b$ are connected in series and then with $\boldsymbol{\lambda}_c$ in parallel. The series and parallel rules (Table~\ref{table_qdet-rules}) yield a deterministic final outcome %
$\boldsymbol{\lambda}=\left(9\left(25+4\sqrt{34}\right)/500, 1-9\left(25+4\sqrt{34}\right)/500\right)\approx\left(0.870, 0.130\right)$, the $G$-concurrence of which is $\approx0.673$.
However, applying a special nondeterministic swapping protocol (the ZZ basis~\cite{QEP-detail_pcalw08}) on $\boldsymbol{\lambda}_a$ and $\boldsymbol{\lambda}_b$, we derive four probabilistic outcomes, $\boldsymbol{\lambda}_1=\boldsymbol{\lambda}_2=\left(81/82,1/82\right)$ and $\boldsymbol{\lambda}_3=\boldsymbol{\lambda}_4=\left(1/2,1/2\right)$, with probabilities $p_1=p_2=0.41$ and $p_3=p_4=0.09$, respectively. The final average $G$-concurrence is $2 p_1 C_d(\mathcal{P}(\boldsymbol{\lambda}_c\otimes \boldsymbol{\lambda}_1))+2 p_3 C_d(\mathcal{P}(\boldsymbol{\lambda}_c\otimes \boldsymbol{\lambda}_3))\approx0.695$. Thus, the series and parallel rules are not optimal. The reason behind this is that the inequality in Eq.~\eqref{theorem_simple-series_5} does not hold in general. For this example, we clearly have $ p_1 \boldsymbol{\lambda}_1^{\downarrow}+ p_2 \boldsymbol{\lambda}_2^{\downarrow}+p_3 \boldsymbol{\lambda}_3^{\downarrow}+ p_4 \boldsymbol{\lambda}_4^{\downarrow}=\left(0.819,0.181\right) \nsucc \mathcal{S}(\boldsymbol{\lambda}_a,\boldsymbol{\lambda}_b)\approx\left(0.966, 0.034\right)$.

\subsection{Series-Parallel}
Finally, for arbitrary series-parallel network topology [Fig.~\ref{fig_seri-para}], we already know from the above that no optimality can be observed for the final average $k$-concurrence. That being said, for qubits  ($d=2$), a special theorem of optimality can indeed be made in terms of the final \emph{worst-case} $G$-concurrence (Appendix~\ref{sec_concurrence}). We prove the following theorem.
\begin{theorem}
	\label{theorem_series-parallel}
	(For $d=2$ only.) Given a series-parallel QN, 
	compared with generic entanglement swapping and concentration protocols of probabilistic outcomes,
	the series and parallel rules (Table~\ref{table_qdet-rules}) produce the optimal worst-case $G$-concurrence between A and B.
\end{theorem}
\begin{proof}
	From Theorems~\ref{theorem_simple-series}~and~\ref{theorem_simple-parallel} we know that, for simple series and simple parallel network topologies, given a generic ensemble of probabilistic outcomes $\alpha$, there must be at least one $\boldsymbol{\lambda}_\alpha$ of which the $G$-concurrence is less than (at most equal to) the $G$-concurrence which the series and parallel rules would produce, respectively. Since $d=2$, this $\boldsymbol{\lambda}_\alpha$ majorizes what the series and parallel rules would produce. Then, given that the series and parallel rules are isotone [Eqs.~\eqref{eq_s_isotone}~and~\eqref{eq_p_isotone}], the final probabilistic outcome related to $\boldsymbol{\lambda}_\alpha$ will remain as the worst case when the series and parallel rules are applied over the full series-parallel network.
\end{proof}
\emph{Remark.---}This special optimality on the worst-case $G$-concurrence for arbitrary series-parallel network topology  with $d=2$ was first discussed in Ref.~\cite{conpt_mgh21}. When $d>2$, Theorem~\ref{theorem_series-parallel} does not hold. This is because other deterministic series or parallel rules different from those in Table~\ref{table_qdet-rules} can exist, and such a rule can sometimes deterministically produce a $\boldsymbol{\lambda}$ with greater $G$-concurrence (which is only possible when the series-then-parallel network topology is involved); hence the worst-case $G$-concurrence is also greater.

{
\subsection{Comparison with Generalized CEP}
\label{sec_cep}
After obtaining a partially entangled state as the DET outcome, it is always possible to further convert it to a maximally entangled state $\boldsymbol{\lambda}_{\max}=\left(1/d,1/d,\cdots,1/d\right)$ with the help of probabilistic entanglement concentration [Eq.~\eqref{eq_concentrate}], of which the success probability represents another useful
{figure of merit} than concurrences for QN. We adopt this figure of merit to directly compare the DET scheme with the benchmark result of the classical entanglement percolation (CEP) (Sec.~\ref{sec_scheme}).
For qubits, it has been proved that deterministic protocols yield higher success probability than CEP~\cite{conpt_mgh21}. However, for general qudits this is not obvious. Here, as our first step, we generalize the idea of CEP to qudits, denoting $p_{\boldsymbol{\lambda}_\to\boldsymbol{\lambda}_{\max}}$ [Eq.~\eqref{eq_concentrate}] as the success probability of converting each QN link $\boldsymbol{\lambda}$ into $\boldsymbol{\lambda}_{\max}$. The overall probability of connecting two distant nodes by a path of maximally entangled states depends on the network topology, which, if series-parallel, allows the probability to be decomposed and calculated by similar series and parallel rules (Table~\ref{table_cep-rules})~\cite{conpt_mgh21}.
\begin{table}[h]
	\centering 
	\caption{\label{table_cep-rules}Series and parallel rules for generalized CEP.\hfill\hfill}
	\begin{tabular}{p{1.5cm}| p{6.7cm}}
		\hline\hline
		Series & 
		$p=p_{\boldsymbol{\lambda}_1\to\boldsymbol{\lambda}_{\max}}p_{\boldsymbol{\lambda}_2\to\boldsymbol{\lambda}_{\max}}\cdots 
        p_{\boldsymbol{\lambda}_n\to\boldsymbol{\lambda}_{\max}}$
		\\
		\hline
		Parallel &  $\left(1-p\right)=\left(1-p_{\boldsymbol{\lambda}_1\to\boldsymbol{\lambda}_{\max}}\right)\cdots 
        \left(1-p_{\boldsymbol{\lambda}_n\to\boldsymbol{\lambda}_{\max}}\right)$
        \\
		\hline\hline
	\end{tabular}
\end{table}
}

{
To compare it with DET, note that for any $\boldsymbol{\lambda}_a$ and $\boldsymbol{\lambda}_b$,
\begin{eqnarray}
\label{eq_cep_1}
p_{\boldsymbol{\lambda}_a\to\boldsymbol{\lambda}_{\max}}\leq p_{\boldsymbol{\lambda}_b\to\boldsymbol{\lambda}_{\max}} \text{ if } \boldsymbol{\lambda}_a\succ \boldsymbol{\lambda}_b.
\end{eqnarray}
Also, for any $p_{\boldsymbol{\lambda}\to\boldsymbol{\lambda}_{\max}}$ [Eq.~\eqref{eq_concentrate}], it is easy to check that
\begin{eqnarray}
\label{eq_cep_2}
\boldsymbol{\lambda}\prec\left(p_{\boldsymbol{\lambda}\to\boldsymbol{\lambda}_{\max}}\right) \boldsymbol{\lambda}_{\max}+\left(1-p_{\boldsymbol{\lambda}\to\boldsymbol{\lambda}_{\max}}\right)\boldsymbol{\lambda}_{\min},
\end{eqnarray}
where $\boldsymbol{\lambda}_{\min}=\left(1,0,\cdots,0\right)$.
Now w.l.o.g.~considering the series rule of DET applied to two links $\boldsymbol{\lambda}_a$ and $\boldsymbol{\lambda}_b$, %
Eq.~\eqref{eq_cep_2} leads to
\begin{eqnarray}
    \label{eq_s_cep}
    \!\mathcal{S}(\boldsymbol{\lambda}_a,\boldsymbol{\lambda}_b)&\prec&
    \mathcal{S}(\left(p_{\boldsymbol{\lambda}_a\to\boldsymbol{\lambda}_{\max}}\right) \boldsymbol{\lambda}_{\max}+\left(1-p_{\boldsymbol{\lambda}_a\to\boldsymbol{\lambda}_{\max}}\right)\boldsymbol{\lambda}_{\min},\nonumber\\
    &&\quad\; \left(p_{\boldsymbol{\lambda}_b\to\boldsymbol{\lambda}_{\max}}\right) \boldsymbol{\lambda}_{\max}+\left(1-p_{\boldsymbol{\lambda}_b\to\boldsymbol{\lambda}_{\max}}\right)\boldsymbol{\lambda}_{\min})\nonumber\\
    &\prec&\left(p_{\boldsymbol{\lambda}_a\to\boldsymbol{\lambda}_{\max}}\right)\!
    \left(p_{\boldsymbol{\lambda}_b\to\boldsymbol{\lambda}_{\max}}\right)\mathcal{S}(\boldsymbol{\lambda}_{\max},\boldsymbol{\lambda}_{\max})\nonumber\\
    &+&\left(p_{\boldsymbol{\lambda}_a\to\boldsymbol{\lambda}_{\max}}\right)\!
    \left(1-p_{\boldsymbol{\lambda}_b\to\boldsymbol{\lambda}_{\max}}\right)\mathcal{S}(\boldsymbol{\lambda}_{\max},\boldsymbol{\lambda}_{\min})\nonumber\\
    &+&\left(1-p_{\boldsymbol{\lambda}_a\to\boldsymbol{\lambda}_{\max}}\right)\! \left(p_{\boldsymbol{\lambda}_b\to\boldsymbol{\lambda}_{\max}}\right)\mathcal{S}(\boldsymbol{\lambda}_{\min},\boldsymbol{\lambda}_{\max})\nonumber\\
    &+&\!\left(1-p_{\boldsymbol{\lambda}_a\to\boldsymbol{\lambda}_{\max}}\right)\!
    \left(1-p_{\boldsymbol{\lambda}_b\to\boldsymbol{\lambda}_{\max}}\right)\mathcal{S}(\boldsymbol{\lambda}_{\min},\boldsymbol{\lambda}_{\min})\nonumber\\
    &=&p_{\boldsymbol{\lambda}_a\to\boldsymbol{\lambda}_{\max}}
    p_{\boldsymbol{\lambda}_b\to\boldsymbol{\lambda}_{\max}}\boldsymbol{\lambda}_{\max}\nonumber\\
    &+&\left(1-p_{\boldsymbol{\lambda}_a\to\boldsymbol{\lambda}_{\max}}
    p_{\boldsymbol{\lambda}_b\to\boldsymbol{\lambda}_{\max}}\right)\boldsymbol{\lambda}_{\min},
\end{eqnarray}
where used in the first step is the isotone property of the swapping function $\mathcal{S}$  [Eq.~\eqref{eq_s_isotone}], used in the second step is convexity (Lemma~\ref{theorem_s_convex}), and used in the last step are the facts that $\mathcal{S}(\boldsymbol{\lambda}_{\max},\boldsymbol{\lambda})=\boldsymbol{\lambda}$ and $\mathcal{S}(\boldsymbol{\lambda}_{\min},\boldsymbol{\lambda})=\boldsymbol{\lambda}_{\min}$. Hence, by Eq.~\eqref{eq_cep_1}, the probability of converting the LHS of Eq.~\eqref{eq_s_cep} to 
$\boldsymbol{\lambda}_{\max}$ is never less than the probability of doing so for the RHS, which is exactly the generalized CEP result $p_{\boldsymbol{\lambda}_a\to\boldsymbol{\lambda}_{\max}}
    p_{\boldsymbol{\lambda}_b\to\boldsymbol{\lambda}_{\max}}$ (Table~\ref{table_cep-rules}).
}

{
For the parallel rule of DET applied to two links $\boldsymbol{\lambda}_a$ and $\boldsymbol{\lambda}_b$, the proof is almost identical by noticing that the concentration function $\mathcal{P}$ has the same isotone [Eq.~\eqref{eq_p_isotone}] and convexity (Lemma~\ref{theorem_d_convex}) properties, and that $\mathcal{P}(\boldsymbol{\lambda}_{\max}\otimes\boldsymbol{\lambda})=\boldsymbol{\lambda}_{\max}$ and $\mathcal{P}(\boldsymbol{\lambda}_{\min}\otimes\boldsymbol{\lambda})=\boldsymbol{\lambda}$. Thus, we have
\begin{eqnarray}
    \label{eq_p_cep}
    &&\mathcal{P}(\boldsymbol{\lambda}_a\otimes\boldsymbol{\lambda}_b)\prec
    \left(1-p_{\boldsymbol{\lambda}_a\to\boldsymbol{\lambda}_{\max}}\right)
    \left(1-p_{\boldsymbol{\lambda}_b\to\boldsymbol{\lambda}_{\max}}\right)\boldsymbol{\lambda}_{\min}\nonumber\\
    &+&\left(p_{\boldsymbol{\lambda}_a\to\boldsymbol{\lambda}_{\max}}+p_{\boldsymbol{\lambda}_b\to\boldsymbol{\lambda}_{\max}}
    -p_{\boldsymbol{\lambda}_a\to\boldsymbol{\lambda}_{\max}}
    p_{\boldsymbol{\lambda}_b\to\boldsymbol{\lambda}_{\max}}\right)\boldsymbol{\lambda}_{\max}.\nonumber\\
\end{eqnarray}
Together, Eqs.~\eqref{eq_s_cep}~and~\eqref{eq_p_cep} indicate that the DET yields higher success probability in terms of both the series and parallel rules. Since both $\mathcal{S}$ and $\mathcal{P}$ are isotones, it is not difficult to extrapolate the proof to any series-parallel QN by iteration, thereby showing that the DEP improves on CEP across arbitrary series-parallel network topology.
}

\section{Quantum Circuit Implementation}
\label{sec_implement}

In this section, we demonstrate the experimental feasibility of the DET scheme using IBM's quantum computing platform \textit{Qiskit}~\cite{qiskit_a21}, showing how to design quantum circuits for the corresponding protocols on qubits ($d=2$).
The performance of the circuits is numerically tested on %
a 5-qubit noisy simulator with its parameters exactly matching the realistic scenario of IBM's quantum hardware. %

Compared with Bell-state-based schemes, the partially-entangled-state-based DET scheme has notable pros and cons: On the one hand, all inputs in DET are only {partially entangled} {and likely easier to generate,} and hence the fidelity will be, in general, higher than using Bell states; on the other hand, the circuit parameters explicitly depend on the inputs, and thus the initial states must be properly estimated by, {e.g.,~heralding~\cite{q-netw_hkmsvtmh18} or tomography~\cite{q-tomogr_bcb21}} before executing the scheme. %
{Since we have shown that any obtained partially entangled qubits can be converted into a singlet for Bell-state-based tasks with higher conversion probability than CEP, 
we believe that our demonstration of DET, representing a thorough study of general operations on partially entangled states, is not only a proof of principle but also practically useful.}

\begin{figure*}[t!]
	\centering
	\includegraphics[height=1.8in]{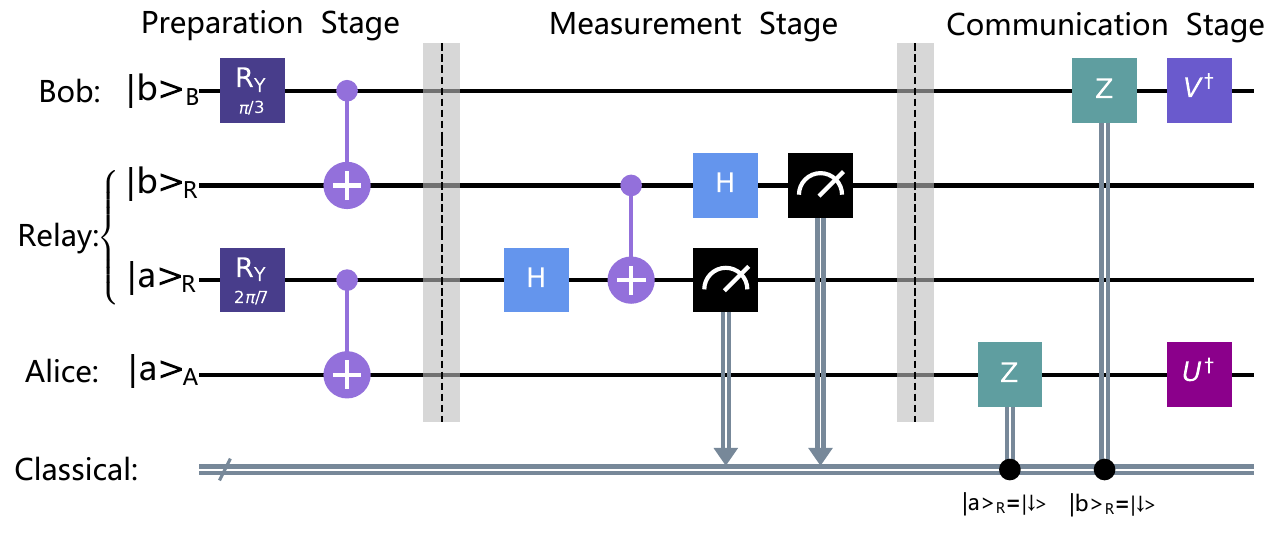}
	\caption{\label{fig_seri_circuit}Deterministic entanglement swapping protocol that implements the series rule (Table~\ref{table_qdet-rules}). Different stages are separated by dashed lines. From left to right: preparation stage (preparing two partially entangled pure states as inputs), %
    measurement stage (transforming into a set of new basis), and communication stage (where Alice and Bob apply local transformations according to the measurement results of the Relay).
		\hfill\hfill}
\end{figure*}

\begin{figure}[t!]
	\centering
	\includegraphics[height=1.6in]{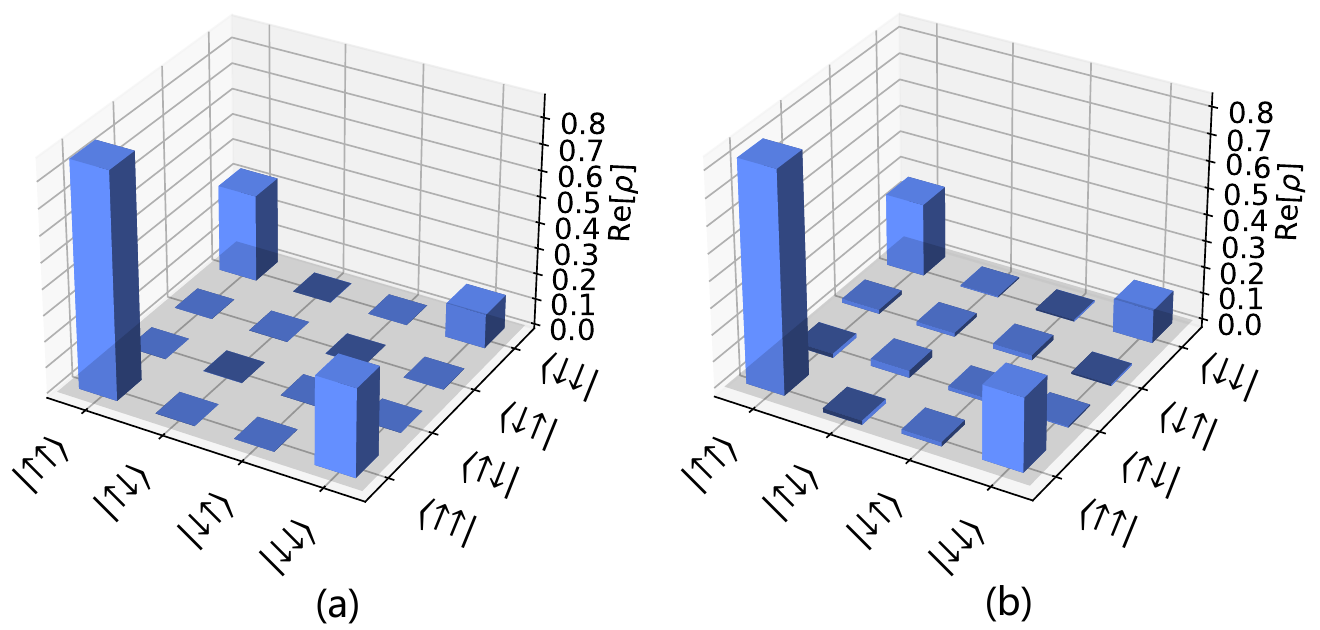}
	\caption{\label{fig_seri_density_matrix}Outcome of the series rule (by Fig.~\ref{fig_seri_circuit}). (a) The noiseless outcome is given by a deterministic pure state $\rho_\text{AB}=\left(0.932\left|\uparrow\uparrow\right\rangle+0.363\left|\downarrow\downarrow\right\rangle\right)\left(0.932\left\langle\uparrow\uparrow\right|+0.363\left\langle\downarrow\downarrow\right|\right)$. 
	(b) The noisy outcome differs from the theoretical value by a fidelity of $92.4\%$, tested on an IBM quantum computation model (``Manila'').
	The imaginary part of $\rho_\text{AB}$ is zero.
		\hfill\hfill}
\end{figure}

\subsection{Series Rule}
The circuit that implements the swapping function $\mathcal{S}(\mathbf{x},\mathbf{y})$ resembles the common design of the Bell-state-based swapping protocol~\cite{entangle-swap_zzhe93}. The main difference is that the Relay (R) measures its qubits in, instead of the Bell basis, a different set of maximally entangled basis,
\begin{eqnarray*}
	\begin{cases}
		\frac{1}{2}\left(+\left|\uparrow\uparrow\right\rangle+\left|\uparrow\downarrow\right\rangle+\left|\downarrow\uparrow\right\rangle-\left|\downarrow\downarrow\right\rangle\right)\\
		\frac{1}{2}\left(+\left|\uparrow\uparrow\right\rangle-\left|\uparrow\downarrow\right\rangle+\left|\downarrow\uparrow\right\rangle+\left|\downarrow\downarrow\right\rangle\right)\\
		\frac{1}{2}\left(+\left|\uparrow\uparrow\right\rangle+\left|\uparrow\downarrow\right\rangle-\left|\downarrow\uparrow\right\rangle+\left|\downarrow\downarrow\right\rangle\right)\\
		\frac{1}{2}\left(+\left|\uparrow\uparrow\right\rangle-\left|\uparrow\downarrow\right\rangle-\left|\downarrow\uparrow\right\rangle-\left|\downarrow\downarrow\right\rangle\right),\\
	\end{cases}
\end{eqnarray*}
which can be implemented by a %
controlled-not (CX) gate
plus two Hadamard gates. 
Here, we use the convention $\left|\uparrow\right\rangle$ and $\left|\downarrow\right\rangle$ to denote a physical qubit's two states, keeping the physical convention distinguishable from the logical convention of the Schmidt numbers.
According to the Relay's measurement results, Alice (A) and Bob (B) choose whether to apply a Z gate to their own qubits, originally prepared as partially entangled with the Relay's qubits given by $\boldsymbol{\lambda}_a=\left(\lambda_{a,1},\lambda_{a,2}\right)$ and $\boldsymbol{\lambda}_b=\left(\lambda_{b,1},\lambda_{b,2}\right)$ respectively, now resulting in a new state $\sqrt{\lambda_{a,1}\lambda_{b,1}}\left|\uparrow\uparrow\right\rangle+\sqrt{\lambda_{a,1}\lambda_{b,2}}\left|\uparrow\downarrow\right\rangle+\sqrt{\lambda_{a,2}\lambda_{b,1}}\left|\downarrow\uparrow\right\rangle-\sqrt{\lambda_{a,2}\lambda_{b,2}}\left|\downarrow\downarrow\right\rangle$ deterministically shared by A and B. The next step is for A and B to apply local unitary transformations $U^\dagger$ and $V^\dagger$ to transform the new deterministic state into a diagonal form with Schmidt numbers $\boldsymbol{\lambda}=\mathcal{S}(\boldsymbol{\lambda}_a,\boldsymbol{\lambda}_b)=\left({\frac{1+\sqrt{1-16\lambda_{a,1}\lambda_{a,2}\lambda_{b,1}\lambda_{b,2}}}{2}},{\frac{1-\sqrt{1-16\lambda_{a,1}\lambda_{a,2}\lambda_{b,1}\lambda_{b,2}}}{2}}\right)$, given by the singular value decomposition (SVD),
\begin{eqnarray*}
	\begin{pmatrix}
		\sqrt{\lambda_{a,1}\lambda_{b,1}}& \sqrt{\lambda_{a,1}\lambda_{b,2}}\\
		\sqrt{\lambda_{a,2}\lambda_{b,1}}& -\sqrt{\lambda_{a,2}\lambda_{b,2}}
	\end{pmatrix}=
	U\text{diag}(\boldsymbol{\lambda})^{1/2}V^\dagger.
\end{eqnarray*}

The full circuit diagram is shown in Fig.~\ref{fig_seri_circuit}, where the input states shared within A--R and R--B are prepared as two partially entangled states, $\boldsymbol{\lambda}_a=\left(\cos^2{\frac{\pi}{7}},\sin^2{\frac{\pi}{7}}\right)$ and $\boldsymbol{\lambda}_b=\left(\cos^2{\frac{\pi}{6}},\sin^2{\frac{\pi}{6}}\right)$, produced by a Y-rotation (RY) gate followed by a CX gate on the four qubits which are initially all in $\left|\uparrow\right\rangle$ (preparation stage).
The noiseless deterministic outcome is $\boldsymbol{\lambda}=\mathcal{S}(\boldsymbol{\lambda}_a,\boldsymbol{\lambda}_b)\approx\left(0.868,0.132\right)$, a partially entangled pure state as expected  [Fig.~\ref{fig_seri_density_matrix}(a)]. 

Testing the circuit on Qiskit, we find that the noisy outcome differs from the theoretical value by a fidelity of $92.4\%$ [Fig.~\ref{fig_seri_density_matrix}(b)]. Note that the mismatch in the outcomes also includes the noise introduced in the preparation stage, which are not part of the swapping protocol. We are convinced that an experimental demonstration of the protocol with similar high fidelity should be possible.

The swapping protocol is relatively mature due to its simplicity, as optimally only one two-qubit gate (CX gate) is needed during the measurement stage. In contrast, the concentration protocol requires more two-qubit gates, as we will see below.

\begin{figure*}[t!]
	\centering
	\includegraphics[height=1.95in]{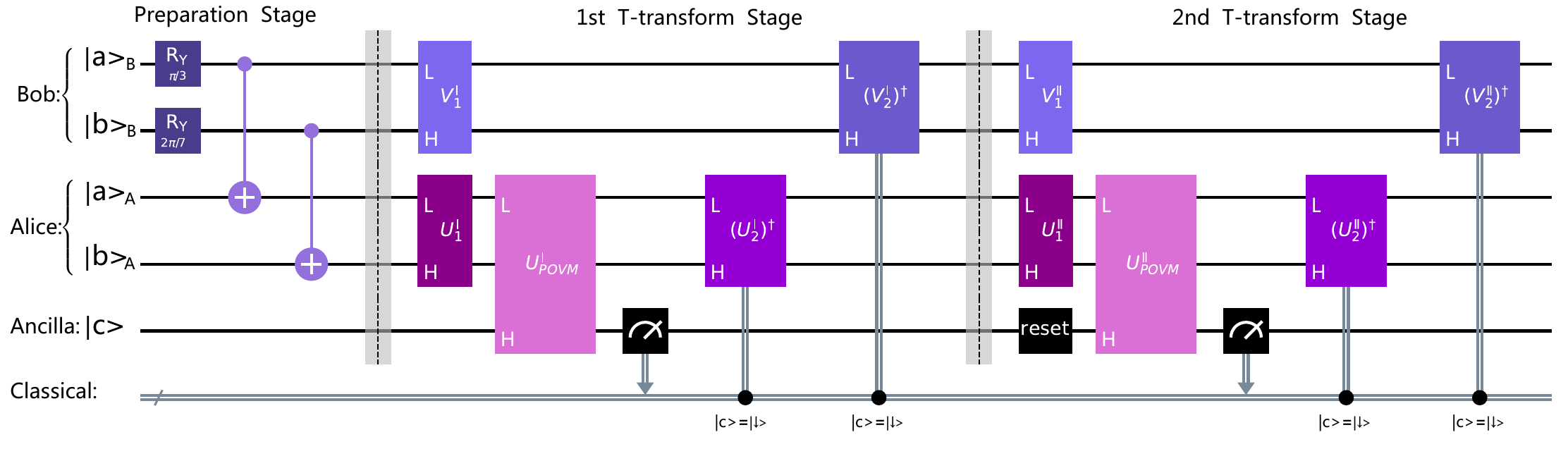}
	\caption{\label{fig_para_circuit}Deterministic entanglement concentration protocol that implements the parallel rule (Table~\ref{table_qdet-rules}). Different stages are separated by dashed lines. From left to right: preparation stage (cf.~Fig.~\ref{fig_seri_circuit}), $1$st T-transform stage, and $2$nd T-transform stage. The final entanglement is concentrated between the qubits $\left|b\right\rangle_\text{A}$ and $\left|b\right\rangle_\text{B}$, which constitute either a partially entangled state or a singlet (in which case a $3$rd T-transform is needed) depending on the initial inputs. 
		\hfill\hfill}
\end{figure*}

\begin{figure}[t!]
	\centering
	\includegraphics[height=1.6in]{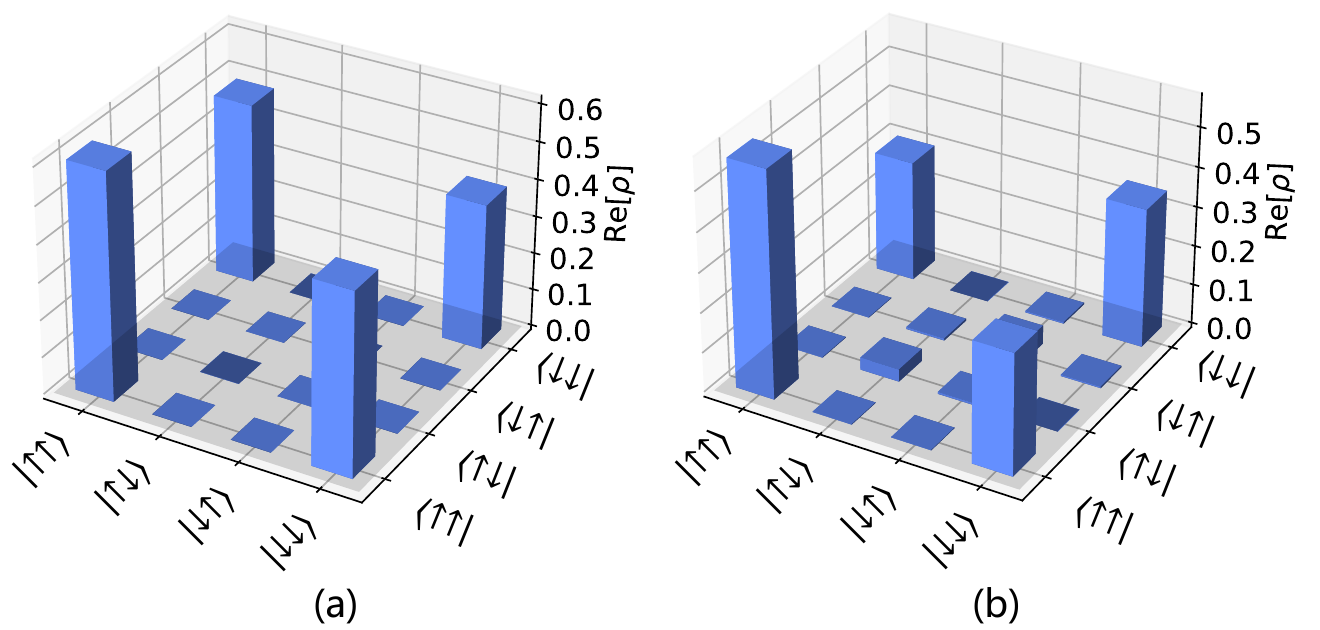}
	\caption{\label{fig_para_density_matrix}Outcome of the parallel rule (by Fig.~\ref{fig_para_circuit}). (a) $\rho_\text{AB}=\left(0.780\left|\uparrow\uparrow\right\rangle+0.625\left|\downarrow\downarrow\right\rangle\right)\left(0.780\left\langle\uparrow\uparrow\right|+0.625\left\langle\downarrow\downarrow\right|\right)$  (cf.~Fig.~\ref{fig_seri_density_matrix}). 
	(b) Fidelity: $78.2\%$.
		\hfill\hfill}
\end{figure}

\subsection{Parallel rule}
The circuit that implements the concentration function $\mathcal{P}(\mathbf{x})$ should deterministically convert the tensor product of the two bipartite states $\boldsymbol{\lambda}_a\otimes\boldsymbol{\lambda}_b=\left(\lambda_{11},\lambda_{12},\lambda_{21},\lambda_{22}\right)\equiv\left(\lambda_{a,1}\lambda_{b,1},\lambda_{a,1}\lambda_{b,2},\lambda_{a,2}\lambda_{b,1},\lambda_{a,2}\lambda_{b,2}\right)$ shared by A and B (w.l.o.g.~$\lambda_{a,1}>\lambda_{b,1}$) to a new state that is equal to (i) $\mathcal{P}(\boldsymbol{\lambda}_a\otimes\boldsymbol{\lambda}_b)=\left(\lambda_{11},1-\lambda_{11}\right)$ when $\lambda_{11}\ge1/2$; or (ii) simply $\mathcal{P}(\boldsymbol{\lambda}_a\otimes\boldsymbol{\lambda}_b)=\left(1/2,1/2\right)$ when $\lambda_{11}<1/2$ [Eq.~\eqref{eq_p}]. Nielsen's original protocol~\cite{nielsen_n99}, based on a series of \emph{T-transforms} to interchange two Schmidt numbers at a time, must be constructed differently for the above two cases. For case (i), two T-transforms are required:
\begin{eqnarray*}
	\left(\lambda_{11},\lambda_{12},\lambda_{21},\lambda_{22}\right)&\overset{\text{T}}{\longrightarrow}&\left(\lambda_{11},\lambda_{12}+\lambda_{22},\lambda_{21},0\right)\\
	&\overset{\text{T}}{\longrightarrow}&\left(\lambda_{11},\lambda_{12}+\lambda_{21}+\lambda_{22},0,0\right).
\end{eqnarray*}
For case (ii), three T-transforms are required:
\begin{eqnarray*}
	&&\left(\lambda_{11},\lambda_{12},\lambda_{21},\lambda_{22}\right)\overset{\text{T}}{\longrightarrow}\left(1/2,\lambda_{12},\lambda_{21},\lambda_{11}+\lambda_{22}-1/2\right)\\
	&&\hspace{15mm}\overset{\text{T}}{\longrightarrow}\left(1/2,\lambda_{11}+\lambda_{12}+\lambda_{22}-1/2,\lambda_{21},0\right)\\
	&&\hspace{15mm}\overset{\text{T}}{\longrightarrow}\left(1/2,\lambda_{11}+\lambda_{12}+\lambda_{21}+\lambda_{22}-1/2,0,0\right).
\end{eqnarray*}
The above arrangements are to guarantee that after each T transform, the Schmidt numbers are still arranged in descending order. 

For simplicity, we only show the circuit diagram for case (i) (Fig.~\ref{fig_para_circuit}). During the first T transform (stage I), A and B begin by applying $U_1^{\text{I}}$ and $V_1^{\text{I}}$ to their qubits, respectively, given by the SVD
\begin{eqnarray*}
	&&\begin{pmatrix}
		\sqrt{\lambda_{11}}& & & \\
		&  \sqrt{\frac{\lambda_{12}+\lambda_{22}}{2}} & & \\
		& & \sqrt{\lambda_{21}} & \\
		& \frac{\lambda_{12}-\lambda_{22}}{\sqrt{2\left(\lambda_{12}+\lambda_{22}\right)}}
		& & \sqrt{\frac{2\lambda_{12}\lambda_{22}}{\lambda_{12}+\lambda_{22}}}
	\end{pmatrix}\\
	&=&
	U_1^{\text{I}}\text{diag}(\left(\lambda_{11},\lambda_{12},\lambda_{21},\lambda_{22}\right))^{1/2}V_1^{\text{I}\dagger}.
\end{eqnarray*}
Next, A applies a POVM (see below)
that probabilistically transforms the above matrix into {either a diagonal-ready matrix, $\text{diag}(\left(\lambda_{11},\lambda_{12}+\lambda_{22},\lambda_{21},0\right))^{1/2}$, or}
\begin{eqnarray*}
	&&\begin{pmatrix}
		\sqrt{\lambda_{11}}& & & \\
		& 0  & & \\
		& & \sqrt{\lambda_{21}} & \\
		& \frac{\lambda_{12}-\lambda_{22}}{\sqrt{\lambda_{12}+\lambda_{22}}}
		& & 2\sqrt{\frac{\lambda_{12}\lambda_{22}}{\lambda_{12}+\lambda_{22}}}
	\end{pmatrix}\\
	&=&
	U_2^{\text{I}}\text{diag}(\left(\lambda_{11},\lambda_{12}+\lambda_{22},\lambda_{21},0\right))^{1/2}V_2^{\text{I}\dagger},
\end{eqnarray*}
which requires extra transformations $U_2^{\text{I}\dagger}$ and $V_2^{\text{I}\dagger}$ applied to A and B before transforming into a diagonal matrix. The two resultant diagonal matrices are identical, indicating that the outcome is deterministic.

Our design of the POVM is as follows. By adding an ancilla qubit initially in $\left|\uparrow\right\rangle$, A can implement the desired POVM by unitarily transforming all her three qubits ($\left|c\right\rangle$, $\left|b\right\rangle_\text{A}$, and $\left|a\right\rangle_\text{A}$) according to %
\begin{eqnarray*}
	U_\text{POVM}^\text{I}=\begin{pmatrix}
		\frac{1}{\sqrt{2}} & 0 & 0 & 0 & -\frac{1}{\sqrt{2}} & 0 & 0 & 0 \\
		0 & 1 & 0 & 0 & 0 & 0 & 0 & 0 \\
		0 & 0 & \frac{1}{\sqrt{2}} & 0 & 0 & 0 & \frac{1}{\sqrt{2}} & 0 \\
		0 & 0 & 0 & 0 & 0 & 0 & 0 & -1 \\
		\frac{1}{\sqrt{2}} & 0 & 0 & 0 & \frac{1}{\sqrt{2}} & 0 & 0 & 0 \\
		0 & 0 & 0 & 0 & 0 & -1 & 0 & 0 \\
		0 & 0 & \frac{1}{\sqrt{2}} & 0 & 0 & 0 & -\frac{1}{\sqrt{2}} & 0 \\
		0 & 0 & 0 & 1 & 0 & 0 & 0 & 0 \\
	\end{pmatrix}
\end{eqnarray*}
and then measuring the ancilla qubit $\left|c\right\rangle$ and recording the result. 
Note that the local unitary transformation $U_\text{POVM}^\text{I}$ is not unique. We choose the above $U_\text{POVM}^\text{I}$ because we find that its most efficient implementation requires only two two-qubit gates, thus easily realizable.

The second T transform (stage II) follows the same procedure, except that the unitary transformations $U_1^{\text{II}}$, $V_1^{\text{II}}$, $U_2^{\text{II}\dagger}$, and $V_2^{\text{II}\dagger}$ are now given by
\begin{eqnarray*}
	&&\begin{pmatrix}
		\sqrt{\lambda_{11}}& & & \\
		&  \sqrt{\frac{\lambda_{12}+\lambda_{22}+\lambda_{21}}{2}} & & \\
		& 
		\frac{\lambda_{12}+\lambda_{22}-\lambda_{21}}{\sqrt{2\left(\lambda_{12}+\lambda_{22}+\lambda_{21}\right)}}
		& \sqrt{\frac{2\left(\lambda_{12}+\lambda_{22}\right)\lambda_{21}}{\lambda_{12}+\lambda_{22}+\lambda_{21}}} & \\
		& & & 0
	\end{pmatrix}\\
	&=&
	U_1^{\text{II}}\text{diag}(\left(\lambda_{11},\lambda_{12}+\lambda_{22},\lambda_{21},0\right))^{1/2}V_1^{\text{II}\dagger}
\end{eqnarray*}
and 
\begin{eqnarray*}
	&&\begin{pmatrix}
		\sqrt{\lambda_{11}}& & & \\
		&  0 & & \\
		& 
		\frac{\lambda_{12}+\lambda_{22}-\lambda_{21}}{\sqrt{\lambda_{12}+\lambda_{22}+\lambda_{21}}}
		&2 \sqrt{\frac{\left(\lambda_{12}+\lambda_{22}\right)\lambda_{21}}{\lambda_{12}+\lambda_{22}+\lambda_{21}}} & \\
		& & & 0
	\end{pmatrix}\\
	&=&
	U_2^{\text{II}}\text{diag}(\left(\lambda_{11},\lambda_{12}+\lambda_{21}+\lambda_{22},0,0\right))^{1/2}V_2^{\text{II}\dagger}.
\end{eqnarray*}
The POVM is given by, accordingly,
\begin{eqnarray*}
	U_\text{POVM}^\text{II}=\begin{pmatrix}	
			\frac{1}{\sqrt{2}} & 0 & 0 & 0 & -\frac{1}{\sqrt{2}} & 0 & 0 & 0 \\
			0 & 1 & 0 & 0 & 0 & 0 & 0 & 0 \\
			0 & 0 & 0 & 0 & 0 & 0 & -1 & 0 \\
			0 & 0 & 0 & \frac{1}{\sqrt{2}} & 0 & 0 & 0 & \frac{1}{\sqrt{2}} \\				\frac{1}{\sqrt{2}} & 0 & 0 & 0 & \frac{1}{\sqrt{2}} & 0 & 0 & 0 \\
			0 & 0 & 0 & 0 & 0 & -1 & 0 & 0 \\
			0 & 0 & 1 & 0 & 0 & 0 & 0 & 0 \\
			0 & 0 & 0 & \frac{1}{\sqrt{2}} & 0 & 0 & 0 & -\frac{1}{\sqrt{2}} \\
	\end{pmatrix}.
\end{eqnarray*}

Given two initial states %
$\boldsymbol{\lambda}_a=\left(\cos^2{\frac{\pi}{7}},\sin^2{\frac{\pi}{7}}\right)$ and $\boldsymbol{\lambda}_b=\left(\cos^2{\frac{\pi}{6}},\sin^2{\frac{\pi}{6}}\right)$
shared between A and B (the same as used in the series rule, Fig.~\ref{fig_seri_circuit}),
the noiseless deterministic outcome is given by %
$\boldsymbol{\lambda}=\mathcal{P}(\boldsymbol{\lambda}_a\otimes\boldsymbol{\lambda}_b)\approx\left(0.609,0.391\right)$,
as expected [Fig.~\ref{fig_para_density_matrix}(a)].

We further decompose all unitary transformations and POVMs into single-qubit gates and CX gates only. 
The overall decomposition produces only $13$ CX gates in total: (1) Decomposing the combination of $V_1^{\text{I}}$, $V_2^{\text{I}\dagger}$, $V_1^{\text{II}}$, and $V_2^{\text{II}\dagger}$ yields only $2$ CX gates, since their POVM-dependent operations are carried out by adjusting the rotation angles of some single-qubit gates only, not involving two-qubit gates. (2) $U_1^{\text{I}}$ requires $2$ CX gates. (3) The specific form of $U_\text{POVM}^\text{I}$ that we choose warrants only $2$ CX gates as well. (4) $U_2^{\text{I}\dagger}$ and $U_1^{\text{II}}$ together require $2$ CX gates [for the same reason as given in item (1)]. 
(5) $U_\text{POVM}^\text{II}$ requires $4$ CX gates, since it is equivalent to adding a CX gate to $\left|b\right\rangle_\text{A}$ and $\left|a\right\rangle_\text{A}$ before $U_\text{POVM}^\text{I}$ and a CX gate after. 
(6) $U_2^{\text{II}\dagger}$ requires $1$ CX gate. 

Requiring $13$ two-qubit gates at maximum, we find that the noisy outcome of the circuit differs from the theoretical value by a fidelity of $78.2\%$ [Fig.~\ref{fig_para_density_matrix}(b)].
Case (ii), however, would be more complicated and require more two-qubit gates, which is not considered here. {Note also that an alternative design of the concentration protocol is available~\cite{entangle-conc_glg06}, which, although arguably simpler to build than Nielsen's~\cite{nielsen_n99}, must use generalized Toffoli gates.}

\section*{Discussion}
What makes the formulation of deterministic entanglement transmission (DET) scalable with network size and adaptable to different (at least series-parallel) network topologies is the use of deterministic quantum communication protocols~\cite{conpt_mgh21}. One may wonder if the specific DET scheme based on the series and parallel rules (Table~\ref{table_qdet-rules}) that we have introduced is the only possible scheme using deterministic protocols. The answer is negative:
For example, the matrix $\mathbf{V}$ in the recipe of $\mathcal{S}(\boldsymbol{\lambda}_a,\boldsymbol{\lambda}_b)$ [Eq.~\eqref{eq_s}] that makes it deterministic is not unique if $\mathbf{V}(\boldsymbol{\lambda}_a,\boldsymbol{\lambda}_b)$ can generally depend on $\boldsymbol{\lambda}_a$ and $\boldsymbol{\lambda}_b$. When properly chosen, the new swapping function can even produce a better $k$-concurrence ($k\neq d$) than our proposed swapping function~\cite{k-concurrence_g05}. Moreover, if quantum catalysts~\cite{q-catal_jp99} are allowed, then it is also possible to have a new concentration function $\mathcal{P}(\boldsymbol{\lambda})$ that is deterministic without satisfying the majorization relation [Eq.~\eqref{eq_p_locc}]. These new functions can be used to build a different set of series and parallel rules.

What remains is the question of whether we can also find deterministic protocols that are scalable for non-series-parallel QN. The complexity of answering this lies in the fact that multipartite protocols~\cite{QEP-GHZ_pclla10} and multilink-based QN routing~\cite{q-netw-route_p19} may have to be used for producing deterministic outcomes. The existence of such deterministic protocols can help us achieve entanglement transmission for a more general QN.

We also find it interesting to apply our scheme to an \emph{infinite} series-parallel QN. %
We expect that percolation-like criticality in terms of $\boldsymbol{\lambda}$ of each link can be observed for entanglement transmission in the thermodynamic limit. However, note that to each QN link we have assigned not a single number but a sequence of Schmidt numbers which have more than one degree of freedom when $d>2$. This prohibits us from establishing an exact one-to-one mapping from $\boldsymbol{\lambda}$ to a real temperature-like parameter. It is unknown how criticality behaves in such case. %

\appendix
\section{Basic Notations}
\label{sec_majorization}
We briefly recall some basic notions in matrix analysis and majorization theory~\cite{matrix-anal}:

Let $\mathbf{x}=\left(x_1,x_2,\cdots,x_n\right)\in\mathbb{R}_{+}^{n}$. Note that some notations for matrices %
can also be defined for $\mathbf{x}$. Given the diagonal matrix of $\mathbf{x}$,
\begin{equation*}
\text{diag}(\mathbf{x})\equiv
\begin{pmatrix}
x_1 & 0&\hdots\\
0 & x_2&\hdots\\
\vdots & \vdots&\ddots
\end{pmatrix},
\end{equation*}
we can simply define the trace of $\mathbf{x}$ as $\tr(\mathbf{x})\equiv \tr(\text{diag}(\mathbf{x}))= \sum_{j=1}^{n}x_j$, the determinant of $\mathbf{x}$ as $\det(\mathbf{x})\equiv \det(\text{diag}(\mathbf{x}))=x_1x_2\cdots x_n$, and
the adjugate of $\mathbf{x}$ as $\mathrm{adj}(\mathbf{x})\equiv %
\left(\det(\mathbf{x})/x_1,\det(\mathbf{x})/x_2,\cdots,\det(\mathbf{x})/x_n\right)$.

Let $\mathbf{x}^{\downarrow}$ and  $\mathbf{x}^{\uparrow}$ be the sequences given by rearranging the coordinates of $\mathbf{x}$ in decreasing order and increasing order, respectively. In other words, $x^{\downarrow}_1, x^{\downarrow}_2,\cdots,x^{\downarrow}_n$ as the coordinates of $\mathbf{x}^{\downarrow}$ satisfy $x^{\downarrow}_1 \ge x^{\downarrow}_2 \ge\cdots\ge x^{\downarrow}_n$. Similarly, for  $\mathbf{x}^{\uparrow}$ there is $x^{\uparrow}_1 \le x^{\uparrow}_2 \le\cdots\le x^{\uparrow}_n$.

Let $\mathbf{x},\mathbf{y}\in\mathbb{R}_{+}^{n}$. We say that $\mathbf{x}$ is \emph{weakly submajorized} by $\mathbf{y}$, or $\mathbf{x}\prec_w \mathbf{y}$~\cite{matrix-anal}, if 
\begin{equation}
\sum_{j=1}^{k}x^{\downarrow}_j \le \sum_{j=1}^{k}y^{\downarrow}_j, \quad \forall k=1,\cdots,n.
\end{equation}
In particular, if $\mathbf{x}$ is weakly submajorized by $\mathbf{y}$ and
\begin{equation}
\sum_{j=1}^{n}x^{\downarrow}_j = \sum_{j=1}^{n}y^{\downarrow}_j,\text{ or, }\tr(\mathbf{x})=\tr(\mathbf{y}),
\end{equation}
then we say that $\mathbf{x}$ is \emph{majorized} by $\mathbf{y}$, or $\mathbf{x}\prec \mathbf{y}$~\cite{matrix-anal}. 

For example, when $\mathbf{x}\in\mathbb{R}_{+}^{n}$ and  $\tr(\mathbf{x})=1$, we always have $\left(1/n,\cdots,1/n\right)\prec \mathbf{x} \prec \left(1,0,\cdots,0\right)$.
On the other hand, given two arbitrary $\mathbf{x},\mathbf{y}\in\mathbb{R}_{+}^{n}$ with $\tr(\mathbf{x})=\tr(\mathbf{y})=1$, we may have both $\mathbf{x}\nprec \mathbf{y}$ and $\mathbf{x}\nsucc \mathbf{y}$. Thus, the majorization relation is not a total order. It is not a partial order either, because $\mathbf{x}\succ \mathbf{y}$ and $\mathbf{y}\succ \mathbf{x}$ do not necessarily imply $\mathbf{x}=\mathbf{y}$, since they may differ by a permutation~\cite{matrix-anal}. The majorization relation is only a \emph{preorder} on $\mathbb{R}_{+}^{n}$.

In this paper, our attention is mostly focused on $d^2$-dimensional normalized states with $d$ Schmidt numbers $\boldsymbol{\lambda}$ satisfying $\tr(\boldsymbol{\lambda})\equiv 1$. We exclusively use the symbol $\boldsymbol{\lambda}$ whenever we implicitly know that it has a \emph{trace of unity} and $d-1$ degrees of freedom. For general vectors, we use the symbols $\mathbf{x},\mathbf{y},\mathbf{z},\cdots\in\mathbb{R}_{+}^{n}$ instead.

\section{The $k$-Concurrence Family of Entanglement Monotones}
\label{sec_concurrence}
Given a $\left(d\times d\right)$-dimensional bipartite pure state, up to unitary equivalence,  $\left|\boldsymbol{\lambda}\right\rangle=\sum_{j=1}^{d}\sqrt{\lambda_{j}}\left|jj\right\rangle$, there are $d$ concurrence monotones, as developed in Ref.~\cite{k-concurrence_g05},
\begin{equation}
C_k(\boldsymbol{\lambda})\equiv\left[{S_k(\boldsymbol{\lambda})}/{S_k(\left(\frac{1}{d},\frac{1}{d},\cdots,\frac{1}{d}\right))}\right]^{1/k}, \quad \tr(\boldsymbol{\lambda})\equiv 1,
\end{equation}
$k=1,2,\cdots,d$, 
where $S_k(\boldsymbol{\lambda})$ is the $k$-th elementary symmetric polynomial~\cite{matrix-anal}, e.g.,
\begin{eqnarray*}
	S_0(\boldsymbol{\lambda})&=&1,\\
	S_1(\boldsymbol{\lambda})&=&\sum\nolimits_{1\le i\le d}\lambda_i,\\
	S_2(\boldsymbol{\lambda})&=&\sum\nolimits_{1\le i<j\le d}\lambda_i \lambda_j,\\
	S_3(\boldsymbol{\lambda})&=&\sum\nolimits_{1\le i<j<k \le d}\lambda_i \lambda_j \lambda_k,\\
	&\cdots&\\
	S_d(\boldsymbol{\lambda})&=&\lambda_1 \lambda_2 \cdots \lambda_d.
\end{eqnarray*}
Note that $C_k$ is named \emph{$k$-concurrence} since the whole monotone family is nothing but a generalization of $C_2(\boldsymbol{\lambda})$, the ``concurrence'' which is more commonly referred to in the context of quantum information theory. Another special concurrence of the family is $C_d$ (i.e.,~$k=d$), the \emph{$G$-concurrence}, as it stands for the \emph{geometric} mean of the Schmidt numbers~\cite{k-concurrence_g05}. The importance of $G$-concurrence can be understood from the main text.

All $k$-concurrences have the following properties~\cite{k-concurrence_g05}.
\begin{itemize}
	\item \emph{Permutation invariance 
 [$\boldsymbol{\lambda}\to\text{perm}(\boldsymbol{\lambda})$]:}
	\begin{equation}
	\label{eq_concurrence_perm}
	C_k(\boldsymbol{\lambda})=C_k(\text{perm}(\boldsymbol{\lambda})).
	\end{equation}
	\item \emph{Unit measure:}
	\begin{eqnarray}
	\label{eq_concurrence_unit}
	0\le C_k(\boldsymbol{\lambda})&\le& 1, \text{ and }\nonumber\\ 
	C_k(\boldsymbol{\lambda})&=& 1  \text{ if } \boldsymbol{\lambda}=\left(\frac{1}{d},\frac{1}{d},\cdots,\frac{1}{d}\right)
	\end{eqnarray}
	\item \emph{Isotone:}
	\begin{equation}
	\label{eq_concurrence_isotone}
	C_k(\boldsymbol{\lambda}_a) \le C_k(\boldsymbol{\lambda}_b) \text{ if } \boldsymbol{\lambda}_a\succ \boldsymbol{\lambda}_b.
	\end{equation}
	\item \emph{Concavity:}
	\begin{equation}
	\label{eq_concurrence_convex}
	C_k(\sum_\alpha p_\alpha \boldsymbol{\lambda}_\alpha)\ge \sum_\alpha p_\alpha C_k(\boldsymbol{\lambda}_\alpha), \quad p_\alpha\in\mathbb{R}_{+}.
	\end{equation}
\end{itemize}
In particular, if $\sum_\alpha p_\alpha=1$, then $p_\alpha$ and $\boldsymbol{\lambda}_\alpha$ can be explained as a probabilistic ensemble. The \emph{average} $k$-concurrence is then defined as the RHS of Eq.~\eqref{eq_concurrence_convex}, $\sum_\alpha p_\alpha C_k(\boldsymbol{\lambda}_\alpha)$, and the \emph{worst-case} $k$-concurrence is defined as $\min_\alpha C_k(\boldsymbol{\lambda}_\alpha)$, accordingly.

\section{Some Useful Theorems}
\label{sec_useful}
\begin{theorem}
	\label{theorem_convex-perm-isotone}
	Let $F:\mathbb{R}_{+}^m\to\mathbb{R}_{+}^n$ be convex, i.e.,
	$\sum_\alpha p_\alpha F(\mathbf{x}_\alpha) \succ F(\sum_\alpha p_\alpha \mathbf{x}_\alpha)$, $p_\alpha\in\mathbb{R}_{+}$. If, for all $\mathbf{x}$, %
	\begin{eqnarray}
	\label{theorem_convex-perm-isotone_1}
	F(\text{perm}(\mathbf{x}))=F(\mathbf{x}),
	\end{eqnarray}
	then $F$ is isotone, i.e.,~$F(\mathbf{x})\succ F(\mathbf{y})$ if $\mathbf{x}\succ \mathbf{y}$.
\end{theorem}
\begin{proof}
	Let $\mathbf{x}\succ \mathbf{y}$ in $\mathbb{R}_{+}^m$. By Theorem II.1.10 of Ref.~\cite{matrix-anal} there exist a set of $m \times m$-dimensional permutation matrices $\mathbf{P}_1,\mathbf{P}_2,\cdots$ and a set of $p_1, p_2,\cdots\in\mathbb{R}_{+}$ with $\sum_{\alpha}p_{\alpha}=1$ such that $\mathbf{y}=\sum_{\alpha}p_{\alpha} \mathbf{P}_{\alpha} \mathbf{x}$. Thus
	\begin{eqnarray}
	\label{theorem_convex-perm-isotone_2}
	F(\mathbf{y})=F(\sum_{\alpha}p_{\alpha} \mathbf{P}_{\alpha} \mathbf{x})\prec \sum_{\alpha} p_{\alpha} F( \mathbf{P}_{\alpha} \mathbf{x})=\sum_{\alpha} p_{\alpha} F(\mathbf{x}),\nonumber\\
	\end{eqnarray}
	i.e.,~$F(\mathbf{x})\succ F(\mathbf{y})$.
\end{proof}

\begin{theorem}
	\label{theorem_sum-product}
	Let $\mathbf{x},\mathbf{y}\in\mathbb{R}_{+}^{n}$. If $\ln \mathbf{x} \succ_w \ln \mathbf{y}$, then
	\begin{eqnarray}
	\label{theorem_sum-product_1}
	\left(\prod_{j=1}^{l-1} x^{\downarrow}_j\right) \sum_{j=l}^{k} \left(x^{\downarrow}_j\right)^s
	\ge
	\left(\prod_{j=1}^{l-1} y^{\downarrow}_j\right) \sum_{j=l}^{k} \left(y^{\downarrow}_j\right)^s,
	\end{eqnarray}
	$\forall k=1,2,\cdots,n$ and $\forall  l=1,2,\cdots,k$, given that $0\le s \le k-l+1$.
\end{theorem}
\begin{proof}
	Let $c=\prod_{j=1}^{l-1} \left({y^{\downarrow}_j}/{x^{\downarrow}_j}\right)\le 1$. Then,
	\begin{eqnarray}
	\label{theorem_sum-product_2}
	\prod_{j=l}^{k} x^{\downarrow}_j\ge c\prod_{j=l}^{k}  y^{\downarrow}_j\ge c^{\left(k-l+1\right)/{s}}\prod_{j=l}^{k} y^{\downarrow}_j=\prod_{j=l}^{k} \left(c^{1/s} y^{\downarrow}_j\right),\qquad
	\end{eqnarray}
	$\forall k=l,l+1,\cdots,n$. Therefore, by Example II.3.5(v) of Ref.~\cite{matrix-anal}, we have
	\begin{eqnarray}
	\label{theorem_sum-product_3}
	\sum_{j=l}^{k} \left(x^{\downarrow}_j\right)^s\ge c \sum_{j=l}^{k} \left(y^{\downarrow}_j\right)^s,
	\end{eqnarray}
	which holds because $f(e^t)=\left(e^t\right)^s$ is convex and monotone increasing in $t$~\cite{matrix-anal}.	
\end{proof}

\section{A Reverse AM--GM Inequality}
\label{sec_reverse}

We will prove the following inequality.
\begin{theorem}
	\label{theorem_reverse}
	Let $\boldsymbol{\varepsilon}\in \mathbb{R}_{+}^{n}$ and $\Delta\ge0$. Let $E_k=\Delta+\sum_{j=1}^{k}\varepsilon_{j}$, $k=1,2,\cdots,n$. If $\varepsilon_{k} \le \varepsilon_{k+1}\le {E_k}/k$, $\forall k$, then
	\begin{equation}
	\label{theorem_reverse-am-gm}
	\left(\frac{E_n}{n}\right)^n \le \varepsilon_{n}\varepsilon_{n-1}\cdots\varepsilon_{1}\left(1+\frac{\Delta/\varepsilon_1}{n}\right)^n
	\end{equation}
	with equality if and only if %
	$\varepsilon_{n}=\varepsilon_{n-1}=\cdots=\varepsilon_{1}$.
\end{theorem}
\begin{proof}
	To start with, note that 
	${E_k}/k={\varepsilon_{k}/k+E_{k-1}}/k\le E_{k-1}/\left(k-1\right)$. Therefore the following inequality holds:
	\begin{equation}
	\label{theorem_reverse-am-gm_1}
	\frac{E_{n}}{n}\le\frac{E_{n-1}}{n-1}\le\cdots\le \frac{E_{1}}{1}.
	\end{equation}
	Now we will prove Eq.~\eqref{theorem_reverse-am-gm} by induction. Suppose
	\begin{equation}
	\label{theorem_reverse-am-gm_2}
	\left(\frac{E_{j-1}}{j-1}\right)^{j-1} \le \varepsilon_{j-1}\varepsilon_{j-2}\cdots\varepsilon_{1}\left(1+\frac{\Delta/\varepsilon_1}{j-1}\right)^{j-1};
	\end{equation}
	we would like to prove 
	\begin{equation}
	\label{theorem_reverse-am-gm_3}
	\left(\frac{E_j}{j}\right)^j \le \varepsilon_{j}\varepsilon_{j-1}\cdots\varepsilon_{1}\left(1+\frac{\Delta/\varepsilon_1}{j}\right)^j.
	\end{equation}
	To do so, define $f_j(\varepsilon_{j},\varepsilon_{j-1},\cdots,\varepsilon_{1})$ as the right-hand side minus the left-hand side of Eq.~\eqref{theorem_reverse-am-gm_3}. Taking the derivative of $f_j$ w.r.t. $\varepsilon_{k}$ for $k=2,\cdots,j$ yields
	\begin{eqnarray}
	\label{theorem_reverse-am-gm_4}
	\frac{\partial f_j}{\partial  \varepsilon_{k}}&=&\frac{\varepsilon_{j}\varepsilon_{j-1}\cdots\varepsilon_{1}}{\varepsilon_{k}}\left(1+\frac{\Delta/\varepsilon_1}{j}\right)^j-\left(\frac{E_j}{j}\right)^{j-1}\nonumber\\
	&\ge&\varepsilon_{j-1}\cdots\varepsilon_{1} \left(1+\frac{\Delta/\varepsilon_1}{j}\right)^j-\left(\frac{E_{j-1}}{j-1}\right)^{j-1}\nonumber\\
	&\ge&\varepsilon_{j-1}\cdots\varepsilon_{1} \left(1+\frac{\Delta/\varepsilon_1}{j-1}\right)^{j-1}-\left(\frac{E_{j-1}}{j-1}\right)^{j-1}\nonumber\\
	&\ge& f_{j-1} \ge 0
	\end{eqnarray}
	with equality if and only if $\varepsilon_{j}=\varepsilon_{j-1}=\cdots=\varepsilon_{1}$.
	Thus, fixing $\varepsilon_{1}$ and noticing that $\varepsilon_{2}$ is constrained by $\varepsilon_{1}$ only, we conclude that, since ${\partial f_j}/{\partial \varepsilon_{2}}\ge0$, $f_j$ takes the minimum if and only if $\varepsilon_{2}\ge \varepsilon_{1}$ actually takes the equality; next, fixing $\varepsilon_{1}$ and $\varepsilon_{2}=\varepsilon_{1}$ and noticing that $\varepsilon_{3}$ is constrained by $\varepsilon_{1}$ and $\varepsilon_{2}$ only, we conclude that $f_j$ takes the minimum if and only if $\varepsilon_{3}\ge \varepsilon_{2}$ actually takes the equality; \ldots. 
	
	Taken together, we conclude that $f_j$ takes the minimum if and only if $\varepsilon_{j}=\varepsilon_{j-1}=\cdots=\varepsilon_{1}$, which put back into $f_j$ yields $f_j(\varepsilon_{1},\varepsilon_{1},\cdots,\varepsilon_{1})= 0$. Hence $f_j\ge 0$.
	The induction is thus completed given $f_1= 0$.
\end{proof}

We note that, despite being a reversed inequality of the AM--GM type, Theorem~\ref{theorem_reverse} is tight. This is because $\Delta$ can take any nonnegative value. The two independent constraints $\varepsilon_{k+1}\ge \varepsilon_{k}$ and $\varepsilon_{k+1}\le {E_k}/k$ %
in Theorem~\ref{theorem_reverse} together imply that the deviations between the coordinates of $\boldsymbol{\varepsilon}$ cannot be too large, which are controlled by $\Delta$. In particular, if $\Delta=0$, then Eq.~\eqref{theorem_reverse-am-gm} just becomes the reverse of the usual AM--GM inequality. However, $\Delta=0$ also requires $\varepsilon_{k+1}=\varepsilon_{k}=\cdots=\varepsilon_{1}$ given the constraints. Therefore Eq.~\eqref{theorem_reverse-am-gm} will not be violated since it can only take the equality.

\begin{corollary}
	\label{theorem_reverse-m}
	(Same prerequisites as in Theorem~\ref{theorem_reverse})
	\begin{equation}
	\label{theorem_reverse-m_1}
	\left(\frac{E_j}{j}\right)^j \le \varepsilon_{j}\varepsilon_{j-1}\cdots\varepsilon_{1}\left(1+\frac{\Delta/\varepsilon_1}{j}\right)^j, \quad j=1,2,\cdots,n
	\end{equation}
	with equality if and only if %
	$\varepsilon_{j}=\varepsilon_{j-1}=\cdots=\varepsilon_{1}$.
\end{corollary}
\begin{proof}
	This has been proved in the proof of Theorem~\ref{theorem_reverse}, which is a special case ($j=n$) of this Corollary.
\end{proof}

\begin{corollary}
	\label{theorem_reverse-d}
	(Same prerequisites as in Theorem~\ref{theorem_reverse})
	\begin{equation}
	\label{theorem_reverse-d_1}
	\left(\frac{E_j}{j}\right)^j \le \varepsilon_{j}\varepsilon_{j-1}\cdots\varepsilon_{1}e^{\Delta/\varepsilon_1}, \quad j=1,2,\cdots,n
	\end{equation}
	with equality if and only if $\Delta=0$.
\end{corollary}
\begin{proof}
	This is because
	\begin{equation}
	\label{theorem_reverse-d_2}
	\left(1+\frac{\Delta/\varepsilon_1}{j}\right)^j\le e^{\Delta/\varepsilon_1}, \quad j>0
	\end{equation}
	with equality if and only if $\Delta=0$.
\end{proof}

\bibliography{ConDET}

\end{document}